\documentclass[prl,onecolumn,showkeys,
nofootinbib,
 amsmath,amssymb,
 aps,
]{revtex4-2}
\usepackage{yfonts}
\usepackage{graphicx}
\usepackage{dcolumn}
\usepackage{bm}
\usepackage{nicefrac}
\usepackage{romannum}
\usepackage{color}
\usepackage[dvipsnames]{xcolor}
\definecolor{darkblue}{RGB}{0,0,196}
\definecolor{darkgreen}{RGB}{0,120,0}
\usepackage[colorlinks=true,linktocpage=true,linkcolor=darkblue,citecolor=red,urlcolor=darkblue]{hyperref}
\usepackage{hyperref}
\usepackage{amsmath,amsfonts, amsthm} 
\usepackage{dsfont}
\usepackage{xcolor}
\usepackage[utf8]{inputenc} 
\usepackage[T1]{fontenc} 
\usepackage{amsmath}
\usepackage{amsfonts} 
\usepackage{amssymb}

\usepackage{graphicx}
\usepackage{caption}
\usepackage{subcaption}

\newcommand{\bea}{\begin{eqnarray}}
\newcommand{\eea}{\end{eqnarray}}
\newcommand{\bel}[1]{\begin{eqnarray}\label{#1}}
\newcommand{\eel}{\end{eqnarray}}

\newcommand*{\myref}[2]{\hyperref[{#1}]{#2}} 

\begin{document}

 \author{Rafał Bistro{\'n}$^{1,2}$,}
 \author{Mykhailo Hontarenko$^{1}$ and}
 \author{Karol {\.Z}yczkowski$^{1,3}$}
 \affiliation{$^1$Institute of Theoretical Physics, Faculty of Physics, Astronomy and Applied Computer Science, Jagiellonian University, ul. Łojasiewicza 11,
30-348 Kraków, Poland}
 \affiliation{$^2$Doctoral School of Exact and Natural Sciences, Jagiellonian University, ul. Łojasiewicza 11, 30-348 Kraków, Poland}
 \affiliation{$^3$Center for Theoretical Physics (CFT), Polish Academy of Sciences,  Al. Lotnik{\'o}w 32/46, 02-668 Warszawa, Poland}

 \title{Bulk-boundary correspondence from Hyperinvariant tensor networks}

\date{Nov. 7, 2024}

\begin{abstract}
We introduce a tensor network designed to faithfully simulate the AdS/CFT correspondence, akin to the multi-scale entanglement renormalization ansatz (MERA), following hyper-invariant tensor network. 
The proposed construction integrates bulk indices within the network architecture to uphold the key features of the HaPPY code, including complementary recovery.
This framework accurately reproduces the boundary conformal field theory's (CFT) two- and three-point correlation functions, while considering the image of any bulk operator.  Furthermore, we provide an explicit methodology for calculating the correlation functions in an efficient manner. 
Our findings highlight the physical aspects of the relation between bulk and boundary within the
tensor network models, contributing to the understanding and simulation of holographic principles in quantum information.
\end{abstract}

                              
\maketitle
\def\bra#1{\left\langle\,#1\,\right|}
\def\ket#1{\left|\,#1\,\right\rangle}
\def\id{\mathds{I}} 
\def\ra{\rangle}
\def\la{\langle}
\def\Tr{\text{Tr}}

\newtheorem{lemma}{Lemma}
\newtheorem{thm}{Theorem}
\newtheorem{cor}{Corollary}
\newtheorem{defn}{Definition}[section]

\newcommand{\raf}[1]{\textcolor{blue}{RB: #1}}
\newcommand{\mh}[2]{\textcolor{purple}{MH: #2}}

\section{Introduction}
Holographic quantum error-correction codes are a modern approach to study the deep connection between the notion of AdS/CFT correspondence, introduced in \cite{Maldacena1998, Witten1998}, and quantum information. 
In recent years, numerous significant ideas have emerged, for a recent review see  \cite{Jahn2021}. Among the most important of these ideas are  HaPPY codes \cite{HaPPy_code}, which rely on the mapping the bulk code space to the boundary by the isometry $\mathcal{V}:\mathcal{H}_{bulk} \xrightarrow{} \mathcal{H}_{boundary}$. The map $\mathcal{V}$ is realized through a network of \textit{perfect tensors} \cite{perfect_tensors}, which can be equivalently identified with absolutely maximally entangled (AME) states \cite{Goyeneche:2015fda,quantum_perfect_tensors}.

Because HaPPY codes are placed on the hyperbolic tiling of the Poincare disk and map bulk degrees of freedom to the boundary ones, they were seen as the first connection between holography and quantum error correction. Moreover, it was proven that HaPPY codes
preserve the area law \cite{ECP10}, 
as the Ryu--Takayanagi formula for the
entropy dependence between bulk and boundary is satisfied \cite{Ryu2006, HaPPy_code}.
Unfortunately, this model is \textit{too} ideal to simulate any physically interesting scenario, because due to the perfect tensor symmetries, any two-point correlations on the boundary become trivial \cite{ChunJun2022},
thus do not depend on the distance between the points analyzed.

From a physical standpoint, the main issue is that perfect tensor networks are too symmetric to encompass spatial directions and to reproduce realistic correlation functions between spatially separated points at the boundary. 
Some approaches to this problem were based on small "distortions" of perfect tensors \cite{bhattacharyya2016exploring, bhattacharyya2018tensor}, in which the first order of two-point correlation's expansion exhibits the desired behaviour.
In contrast entirely different tensor networks, without bulk indices, called MERA, can indeed simulate the conformal field theory by producing correct two- and three-point correlation functions \cite{Pfiefer2009}. Moreover, within this approach, one can obtain the central charge of the Virasoro algebra
\cite{Holzhey1994,Pfiefer2009},
using the von Neumann entropy in the continuum limit.

Motivated by MERA networks,  Evenbly introduced  the concept of hyper-invariant tensor networks (HTN) \cite{Evenbly2017}. This pioneering work from 2017 studied the entanglement properties of HTNs and introduced a suitable discrete counterpart of scaling, thereby showing a connection with CFT. Using this technique various quantities such as the central charge 
$c$ and scaling dimensions of primary fields were computed \cite{ChunJun2022, Jahn2021}.
An alternative approach \cite{Jahn2022}, aims to create CFT directly on the tensor network. This method extends the continuum notion of CFT to a discrete version, called Quasi-CFT (qCFT). 
Following this direction, the work \cite{Steinberg2022} proven that the construction of MERA-type tensor networks on hyperbolic tiling is at least a close approximation of continuum CFT, which is considered an example of qCFT \cite{Jahn2022}.
However, to the best of our knowledge, no previous construction captured bulk-boundary mapping and presented a rigorous discussion of correlations on the boundary.

The aim of this work is to introduce a 
tunable modification of the 
HaPPY code, which decreases the degree 
of its symmetry,
but leads to realistic correlation functions.
The proposed model connects them with nodes of
hyper-invariant tenor networks \cite{Evenbly2017, Steinberg2022},
constructed from dual unitary matrices
\cite{AWGG16,BKP19,BRRL24}.
Such a technique allows us to combine the benefits of a HaPPY code -- bulk to boundary mapping, with the hyper-invariant ones -- desired forms of correlation functions. 

The proposed construction satisfies the definition of block-perfect tensors \cite{harris2018calderbanksteaneshor},
which form a particular case of planar $k$-uniform tensors \cite{planar_k_uniform_states}.
However, as the present work was motivated by hyper-invariant tenor networks
and hyper-invariant codes \cite{Steinberg2023},
we prefer to stick to this nomenclature.

Within the presented approach, we managed to prove the desired behavior of two- and three-point correlation functions in large network limits, while considering the image of any bulk operator.
To derive these results we developed a methodology to conveniently and explicitly calculate desired correlation functions in an efficient manner.
Due to the freedom in choosing the building blocks in our construction, we managed to numerically scan a wide range of scaling dimensions $\Delta$, demonstrating the potential to tailor our model to a specific value, corresponding to physically interesting fields.

A  HaPPY code \cite{HaPPy_code} is determined by the geometry of the network and is applicable
for any dimension $D$, provided a perfect tensor, with $n$ indices running from $1$ to $D$
does exist, which is equivalent to existence of AME$(n,D)$ states  \cite{HESGW18,HW19b}.
It is often sufficient to take $D=2$ and work with qubits. 
For instance, the original code \cite{HaPPy_code} is based on a perfect tensor with six indices each running from 1 to 2,
which corresponds to an AME state of six qubits.
However, choosing a larger local dimension $D$ 
significantly increases the space of possibilities.
As the construction of hyper-invariant tensor network \cite{Evenbly2017} employ at least pairs of qubits $d^2 = 4$,
we present a suitable network also in the setting $D = d^2 = 4$. However, any other local dimension $d$ is suitable as well.

The paper is organized as follows.
In Section 2, we present the construction of a proposed tensor and discuss in detail the properties of each of its components.
Section 3 is devoted to physical properties streaming from the obtained tensor network. We start by mentioning the bound on a central charge of the conformal field theory mimicked in this way. Next, we focus on the main objective of this paper -- the correlation functions. While calculating correlation functions,  we thoroughly discuss a reduction of the tensor network and prove their desired forms in the limit of a large network size. 
Finally, in Section 4, we investigate an available range of scaling dimensions numerically.
In Section 5 we give final remarks and formulate some natural open problems.

In Appendix A we present a complete recipe for the construction of the proposed tensor network and discuss explicit examples, while in Appendix B we recall the precise definitions of hyper-invariant tensor networks and Evenbly codes.
The detailed discussion of double unitary matrices, serving as building blocks of the hyper-invariant part of the proposed construction, is presented in Appendix C.
In Appendix D we discuss a scaling dimension in the general case.
The final Appendix E concerns the numerical calculation of the three-point correlation function in large network limit.

\section{Tensor networks construction}\label{Sec:Node}

In this Section, we present a new construction of a network realizing the isometry between the Bulk and Boundary Hilbert spaces. Similarly, as well-established HaPPY codes \cite{HaPPy_code} and hyper-invariant tensor networks \cite{Evenbly2017}, we place our structure on a tessellation of Poincare disc, which corresponds to a constant time slice of Anti-de-Sitter (AdS) space of $2+1$ dimension. We choose standard $\{5,4\}$ tiling of regular pentagons, four of which meet at each vertex, see Fig. \ref{fig:Fig3} (b). This frame is filled, layer by layer, with tensors with one bulk index pointing "up" and five boundary indices, which contract with neighboring tensors. The free boundary indices in the last filled layer correspond to local subspaces of the boundary Hilbert space.

The motivation behind the construction of the proposed tensor is straightforward. We want to combine a perfect tensor $T$ with a hyper-invariant invariant tensor $F$, thus "reducing" the symmetries of a perfect tensor network in a minimal manner, necessary to reproduce key features of CFT.
A perfect tensor $T_{ijklmn}$ is such a tensor, that every permutation and grouping of its indices into two sets, ex. $V_{kl;\,ijmn} := T_{ijklmn}$, gives an isometry $V$.  

In order to reconstruct the physical properties of AdS space, we impose additional "rotation symmetry" conditions for a perfect tensor --  tensor should remain invariant under cyclic change of tensor indices, except the first, bulk one,  i.e. $T_{ijklmn} = T_{injklm}$ -- see Fig. \ref{fig:Fig1}(a). The symmetry of translations corresponds to the requirement for all nodes in the tensor network to be identical.
As an exemplar tensor of this form, we chose a perfect tensor of $6$ quarts ($D = d^2= 4$), corresponding to AME$(6,4)$ state, constructed with three orthogonal Latin hypercubes of order four \cite{Goyeneche:2015fda}, since its local preprocessing by permutation matrices is sufficient to obtain desired form.

The second element of our construction "frame" $F_{abcde}$ is inspired by a node of a hyper-invariant tensor network. Its key property is being planar $2$-uniform \cite{planar_k_uniform_states}: its contraction with its hermitian conjugate on $3$ neighbor pair of indices reduces it to identities on remaining pairs.
Since the exact definition of a hyper-invariant tensor network is quite complex, and not essential in further discussion we presented in the Appendix \myref{App:Hyperinv}{B}.
We propose a construction of such tensor in the spirit of \cite{Evenbly2017} with five copies of the dual-unitary matrix, see Figure \ref{fig:Fig1} (b) and (c), which is why we call it hyper-invariant.
To guarantee desired reductions of
hyper-invariant tensor frame we need only two simple identities, corresponding to two types of orthogonalities, satisfied by dual-unitary matrices, presented in Figure \ref{dual_uni}. Matrix with such properties can be realized on pair of qubits -- ququart ($d^2 = 4$), see example in following sections.
\begin{figure}[h!]
\centering

\begin{subfigure}[h]{0.16\textwidth}
    \subcaptionbox{}{
        \includegraphics[width=1.5in]{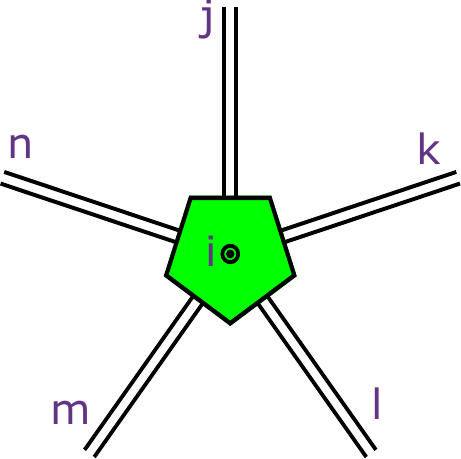}
    }
\end{subfigure}%
\hspace{0.5cm}
\begin{subfigure}[h]{0.3\textwidth}
    \subcaptionbox{}{
        \includegraphics[width=2.5in]{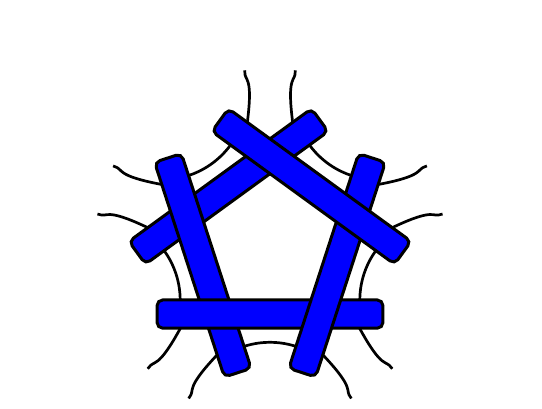}
    }
\end{subfigure}%
\hspace{1cm}
\begin{subfigure}[h]{0.3\textwidth}
    \subcaptionbox{}{
        \includegraphics[width=2in]{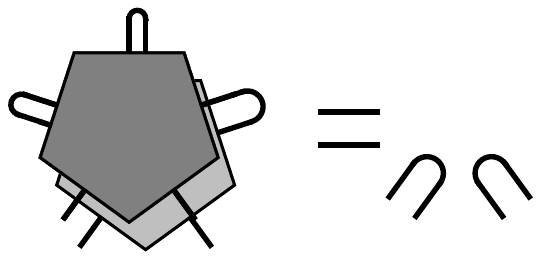}
    }
\end{subfigure} 
\caption{\textbf{(a)} Diagrammatic representation of a perfect tensor
of six subsystems of dimension $D=4$
obtained from $3$ orthogonal Latin hypercubes of order four
-- see Fig. \ref{hyper_orto_fig}.
The first index $i$, pointing in the "bulk direction",
is denoted as a dot, while all the others indices, $j,\dots,n$,
point in the "boundary" direction.
Following example from the text, each line represents a single qubit, 
so a ququart index is represented by a double line.
Appropriate permutation matrices applied on outgoing indices make this tensor 'rotationally invariant'
with respect to rotation of the pentagon by $2\pi/5$. 
\textbf{(b)} Hyper-invariant tensor "frame" constructed from five  unitaries  of order $D=4$ (blue rectangles). 
\textbf{(c)} Desired contraction rule of hyper-invariant frame \cite{planar_k_uniform_states}, pale color corresponds to conjugation.  Contraction of this tensor with its Hermitian conjugate on $3$ neighbor pair of indices reduces it to identities.}
\label{fig:Fig1}
\end{figure}

\begin{figure}[h]
\centering
\begin{subfigure}[h!]{0.4\textwidth}
\centering
 \subcaptionbox{}{
        \includegraphics[width=1.5in]{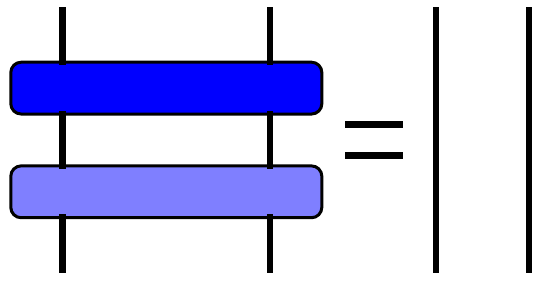}
    }
\end{subfigure}%
~
\begin{subfigure}[h!]{0.55\textwidth}
\centering
 \subcaptionbox{}{
        \includegraphics[width=3in]{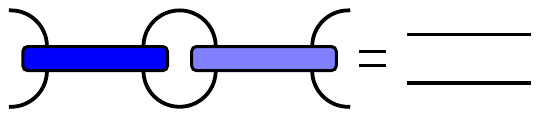}
    }
\end{subfigure} 
\\
\caption{Necessary properties for blue dual-unitaries. The pale colour corresponds to hermitian conjugation.}
\label{dual_uni}
\end{figure}

The final step in our model is to combine perfect tensor $T$ with hyper-invariant frame $F$ into one node $N$ by entangling them on outgoing indices with large unitaries, as presented in Figure \ref{fig:Fig3} (b).
This step is crucial because otherwise the entire tensor network would split up into two uninteresting networks, the first consisting of perfect tensors would be just a HaPPY code, and the second, hyper-invariant one, would have no connection to the bulk.
In the discussed example these unitaries act on pairs of ququarts, $D^2=d^4 = 16$.  

Thus we presented a simple recipe of a hyper-invariant tensor -- one network node, together with an example with local bulk dimension $d^2 = 4$ and boundary dimension $d^4 = 16$. From now on we will refer to all outgoing qubit indices in one direction as just one index.
The explicit construction of the proposed tensor network, together with examples is presented in  Appendix \myref{App:Tensor_node}{A}.

Constructed nodes have profitable reduction rules: The result of contracting the tensor with its conjugate on three neighboring boundary indices is an identity on all remaining pairs of indices.
This allows us to build a hyper-invariant tensor network, with appropriate reduction rules, but with possibly non-vanishing  correlation functions on the boundary. 

Throughout the paper, we consider a tensor network built in a vertex inflation manner. In the zeroth step, the network consists of only one node. In each consecutive step new tensor is added to the free boundary index of the previous layer and then an additional tensors are added to close each vertex with already three tensors, as presented in Figure \ref{fig:Fig3} (b).
While constructing a tensor network in such a regular way, each new tensor, when added, has at least $3$ neighboring non-contracted boundary indices, which is more than necessary to serve as isometry from bulk and $2$ contracted boundary indices.  
Thus the proposed construction is indeed an isometry from bulk into the boundary, by the same token as perfect tensors are.
Using the same arguments one can see that we chose the simplest possible tiling suitable for our construction. If the basic tile would be a square, instead of the pentagon, then the node which closes a loop around a vertex has only $2$ free boundary indices out of $4$ and the bulk one, thus it cannot be an isometric by itself, since it maps from the larger space into a smaller one. 
In our discussion we omit the annoying normalization of a hyper-invariant tensor as an isometry, to later restore it, if necessary. 

\begin{figure}[h!]
\centering
\begin{subfigure}[h!]{0.5\textwidth}
\centering
 \subcaptionbox{}{
        \includegraphics[width=3in]{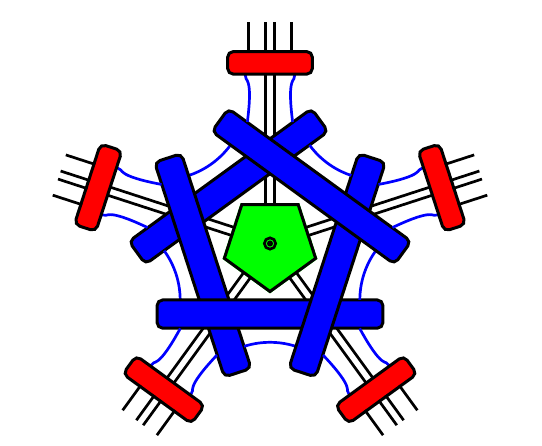}
    }
\end{subfigure}%
~
\begin{subfigure}[h!]{0.5\textwidth}
\centering
 \subcaptionbox{}{
       \includegraphics[width=3in]{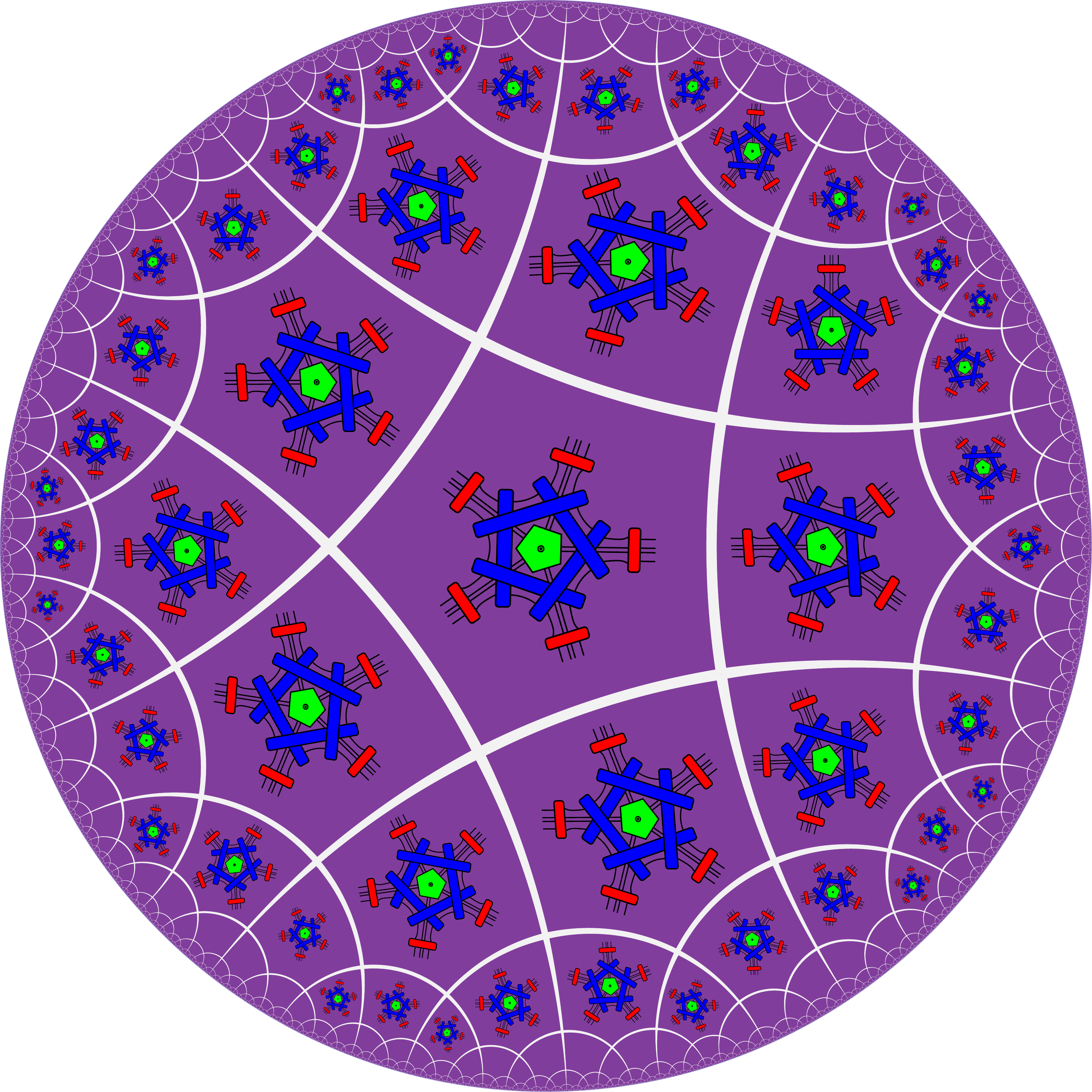}
    }
\end{subfigure} 
\\
\caption{
\textbf{(a)} Entire construction of a
hyper-invariant tensor. Green pentagon corresponds to a 
perfect tensor with six indices running from 1 to $D=4$,
 blue rectangle denote  dual-unitary matrices of order $D$, 
 forming a hyper-invariant frame,
while red unitaries of size $D^2=16$ entangle two elements in 
a single tensor node. 
Each line represents a single qubit,
 blue lines represent indices on which dual-unitaries act. 
 The bulk index, pointing "inward", is denoted 
 as a dot inside a green pentagon.
\textbf{(b)} Example of a constructed tensor network with $2$ layers and the central node. }
\label{fig:Fig3}
\end{figure}

Constructed network based on a perfect tensor $T$,
a with hyper-invariant frames $F$ as nodes and pairs of unitaries on edges presented in Fig. \ref{fig:hyper} leads to a relaxed 
version of an Evenbly code \cite{Steinberg:2024ack} --
see Appendix \myref{App:Hyperinv}{B}. 
An analogous construction without perfect tensors 
forms a direct realization of a hyper-invariant tensor network \cite{Evenbly2017}, presented in Appendix \myref{App:Hyperinv}{B}.

\begin{figure}[h!]
    \centering
    \includegraphics[width=0.6\linewidth]{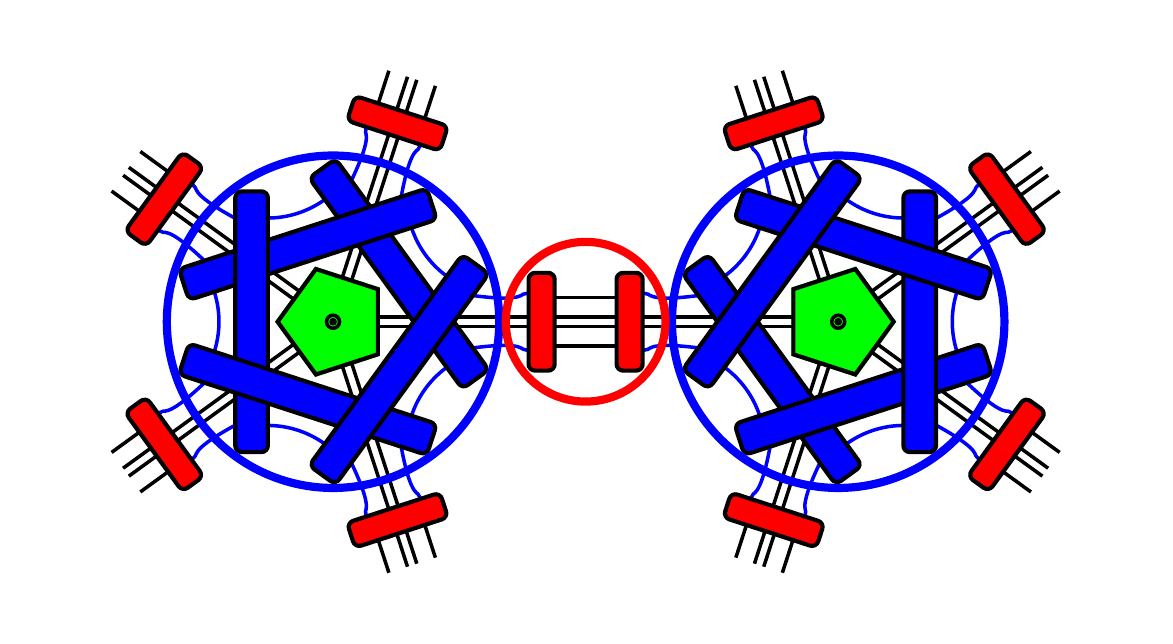}
    \caption{Connection of $2$ nodes from the tensor network. A pair of unitaries (rectangles in red) from neighboring nodes can be considered as an edge tensor (red circle) and a perfect tensor with a frame as a node tensor (blue circle) from hyper-invariant code.}
    \label{fig:hyper}
\end{figure}

The perspective of hyper-invariant codes naturally indicates gauge freedom in a tensor network constructed from presented nodes. 
Denoting entangling unitary gate (red rectangle) as $U$, 
we obtain the product $UU^\top$ at each edge of the tensor network
-- see Fig \ref{fig:hyper}. Replacing $U$ by its product with any orthogonal matrix $U \to U Q$, does not change the  behavior
of the network  as $(UQ)(UQ)^\top = UQQ^\top U^\top = UU^{\top}$. 
One needs only o take into account the uncontracted 
boundary indices on which the action of $Q$ does not disappear.

\section{Properties of construction proposed}

We start this section by mentioning the central charge of conformal field theory, which can be deduced from tensor network construction. Then we proceed to discuss the heart of conformal theories - the correlation function.
The main goal for HaPPY code and its successors is to provide a tool reproducing AdS/CFT correspondence properties. Since both of those theories are field theories, it is natural to consider the mapping of \textit{operators} instead of states. We follow this approach throughout this section.

One of the main properties of interest in conformal field theory is the central charge $c$ of its conformal algebra. It is a quantum anomaly, arising in quantizations of Witt algebra, describing conformal transformations, into Virasoro algebra. 
For any conformal theory
the commutation relations between generators $L_i$ reads,  

$$[ L_m, L_n ] = (m-n)L_{m+n}+\frac{c}{12}m(m-1)^2\delta_{n+m}.$$
Because conformal theories are naturally based in continuous geometry settings, it is not yet well understood, whether we can consistently speak about the Virasoro algebra in the discrete setting. Nevertheless, some attempts were made to construct the discrete theories, called qCFT \cite{Jahn2022_1}.
Fortunately, the work \cite{Holzhey1994} established a relation between central charge, geodesic length, and von Neumann entropy for any conformal field theory. This allows us to compute the central charge by the discrete version of conformal scaling -- adding layers to the tensor network, and computing the difference between entropy. The discussion and application of this approach for a hyperbolic tilling of Poincare disc (corresponding to AdS (2+1) space) can be found in \cite{Evenbly2014, Evenbly2011, Jahn:2019nmz, Jahn2022}.  Using Eq. (23) from  \cite{Central_charges_and_scaling}  we provide an upper bound for the central charge in our theory:
\begin{equation}
\label{central_estimation}
c_{\{5,4\}_{v}} \leq \frac{9 \ln{16}}{\ln(2 + \sqrt{3})  } \approx 18.95.
\end{equation}
The subscript $\{5,4\}$ is the Schl{\"a}fli symbol that describes the tiling of the Poincare disk, whereas $v$, corresponds to the vertex inflation method of adding a new layer to our network.

We stress the inequality sign, which results from the extrapolation of entanglement between each neighboring pair nodes to be maximal. This property was proven for the HaPPY codes \cite{HaPPy_code}, however, in the case of hyper-invariant tensors, up to our knowledge, it is still a hypothesis. Furthermore, our numerical calculations suggest that the actual value of the central charge $c_{\{5,4\}_{v}}$ is close to the upper bound \eqref{central_estimation}.
A large value of the central charge provides an argument, that our model does not encompass the so-called minimal conformal field theories, for which $0 < c \leq 1$ \cite{Ginsparg:1988ui}. 

We start the discussion of correlation functions by invoking a well-known fact, that conformal invariance of conformal field theories fixes the structure of two- and three-point correlation functions \cite{Ginsparg:1988ui}.
Let $\phi(x)$ be to quantum field acting on states in corresponding Hilbert space. Let $\ell _{ij}$ be a distance between two chosen points calculated in conformal theory -- along the boundary of AdS space. Thus, if the theory is conformal invariant, the correlation of field in points $x_1$ and $x_2$ is given by:
\begin{equation}
\label{eq:2point_corr_behaviour}
 \la \phi(x_1)\phi(x_2) \ra \sim \frac{1}{\ell_{12}^{2\Delta}},
\end{equation}
where the proportionality constant can be set to $1$ by normalization of fields. The parameter $\Delta$ is called a scaling dimension for a field and characterizes its transformation properties.
The three-point correlations on the other hand are given by
\begin{equation}
\label{3point_corr_behaviour}
 \la \phi(x_1)\phi(x_2)\phi(x_3) \ra = \frac{C_{123}}{\ell_{12}^{\Delta} \,\,\, \ell_{23}^{\Delta} \,\,\, \ell_{13}^{\Delta}},
\end{equation}
where constants $C_{123}$ are field-dependent parameters of a given conformal theory.

To demonstrate how the correlations arise in the discussed model of AdS/CFT correspondence, we outline the next procedure for mapping operators within a tensor network.
Let $T^{i_{1},\cdots,i_{n},j_1,\cdots, j_k}$ be a tensor network, constructed from multiple nodes, where the indices $i_1,\cdots,i_n$ correspond to the bulk indices of each tensor within the network and $j_1,\cdots,j_k$ to the uncontracted boundary indices of the outer layer on nodes. 
The enumeration of indices in the tensor network is only used to provide an explicit formula, and it does not affect the results, thus we let ourselves not discuss it further. 
 
Furthermore, let $\mathcal{O}_B$ be an operator, for clarity localized in some point in the bulk,
represented by matrix ${\mathcal{O}_B}_{\,\,\,\, i_{r}}^{i_{r}'}$. The mapping of this operator to boundary states corresponds to  "sandwiching" it between the tensor network and its conjugation, which can be combined with partial trace on uninteresting subsystems as presented in the example below which leaves only two boundary subsystems
\begin{equation}
    T^{i_{1},\cdots, i_r, \cdots,i_{n},j_1,\cdots,j_l,j_m,\cdots, j_k} \overline{T}_{i_{1},\cdots,i_r' ,\cdots,i_{n},j_1,\cdots,j_l',j_m',\cdots, j_k}{\mathcal{O}_B}_{\,\,\,\, i_{r}}^{i_{r}'}= \phi_{\,\,\,\,\,\,\,\,j_{l}' j_{m}'}^{j_{l}j_{m}}~,
\label{eq:mapping}
\end{equation}
with Einstein's summation convention over repeated index assumed.

The expressions like \eqref{eq:mapping} can be interpreted as tensor networks as well and represent the image of a bulk operator on a boundary. We note that the operator $\mathcal{O}_B$ in general does not need to be local.
In order to obtain correlation functions, we chose some probing observables, in this example two, denoted by ${v_1}_{\,\,\,\,j_l}^{j_l'}, {v_2}_{\,\,\,\,j_m}^{j_m'}$, with indices corresponding to points $x_{1}$ and $x_{2}$ on the boundary and calculate their joint expectation value. Thus we have

\begin{equation}
  \la \phi(x_{1}) \phi(x_{2}) \ra :=     \phi_{\,\,\,\,\,\,\,\,j_{l}' j_{m}'}^{j_{l}j_{m}} {v_1}_{\,\,\,\,j_l}^{j_l'} {v_2}_{\,\,\,\,j_m}^{j_m'},
\end{equation}
which could be geometrically viewed as in Fig.
\ref{fig:twopoint} (a).

\begin{figure}[h]
\centering
\begin{subfigure}[h!]{0.5\textwidth}
\centering
 \subcaptionbox{"Sandwiching" }{
        \includegraphics[width=3in]{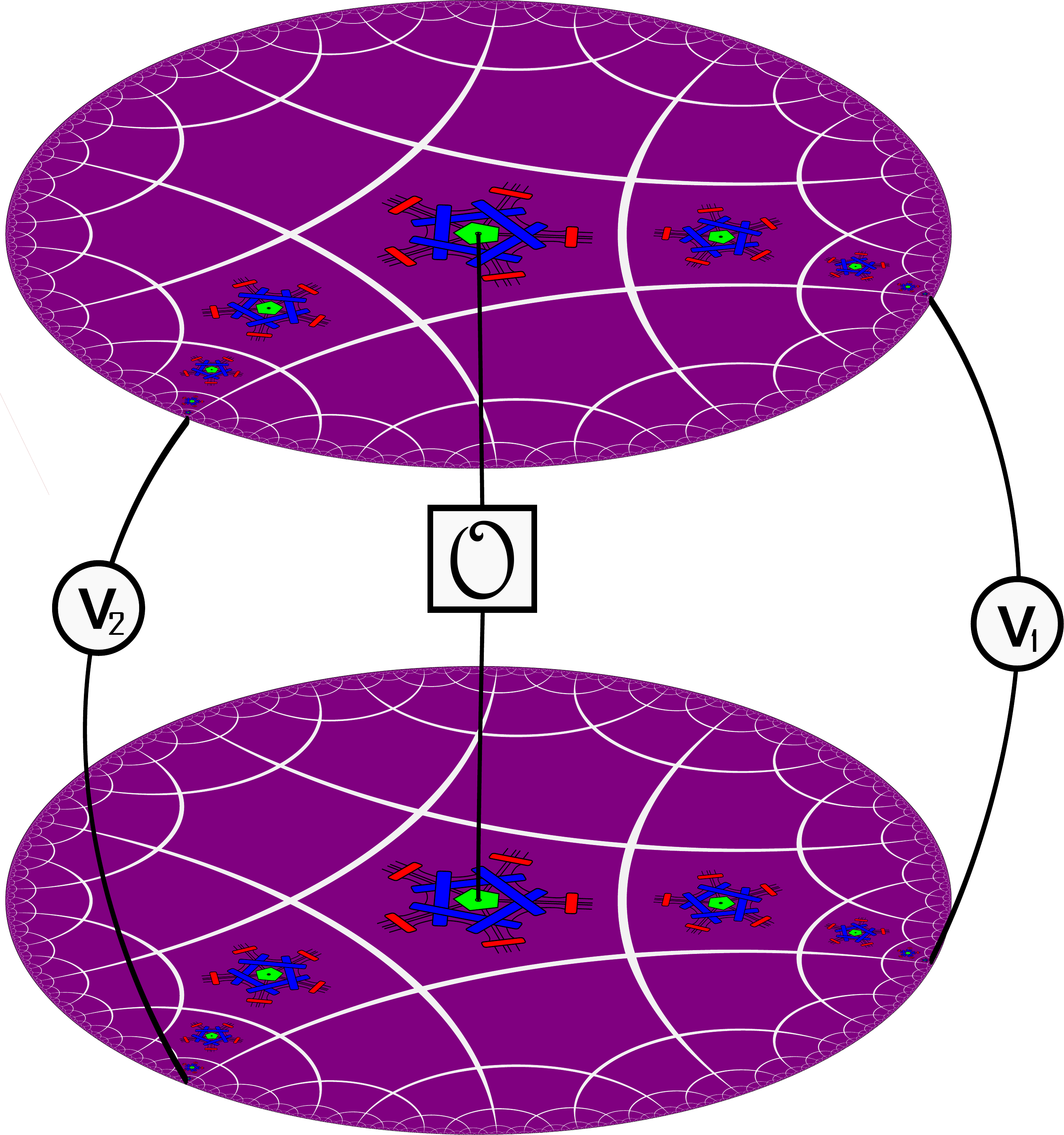}
    }
\end{subfigure}%
~
\begin{subfigure}[h!]{0.5\textwidth}
\centering
 \subcaptionbox{}{
       \includegraphics[width=3in]{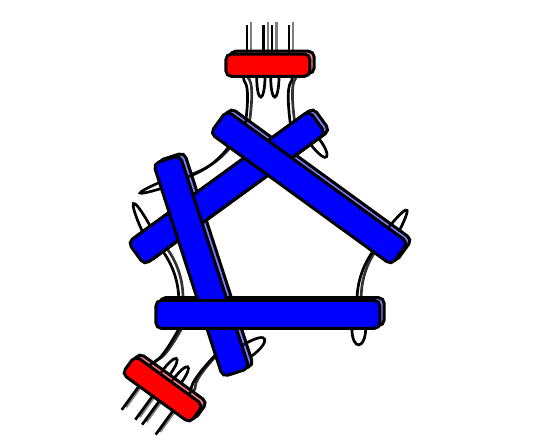}
    }
\end{subfigure} 
\\
\caption{\textbf{(a)} Geometric representation of mapping the bulk operator $\mathcal{O}_B$ to the boundary region, where contracted indices, represented as lines, specify the position of the mapped operator $\mathcal{O}_B$ and the probing observables $v_1, v_2$.\newline
\textbf{(b)} The node from the path after all possible reductions, the conjugations of each remaining matrix are denoted as their pale copies placed in the back. Note the simplification of the perfect tensor irrespective of bulk operators.}
\label{fig:twopoint}
\end{figure}

Hereafter for clarity, we switch to index-free notation, in which the "sandwich" of some bulk operator $\mathcal{O}_B$ between the tensor network and its conjugation, combined with partial traces and expectation values of probing operators can be written as:
\begin{equation*}
\la \phi_1(x_1) \phi_2(x_2) \cdots\ra = \Tr\left[\mathcal{V} \mathcal{O}_B \mathcal{V}^\dagger\left(v_1 \otimes \id \otimes v_2 \otimes \id \cdots\right) \right],
\end{equation*}
with $\mathcal{V}$ represent the bulk-boundary isometry of the tensor network. 

Moreover, we will focus on probing operators with zero traces. To justify this choice, let us consider a one-point correlation function. Note, that while taking trace one can, in the first step reduce all tensors with $3$ boundary indices, resulting in a new boundary to reduce in the consecutive step. Repeating this procedure, by the same token as for perfect tensors, one obtains a very simple form of a one-point correlation function,
\begin{equation}
\la \phi(x_1)\ra \ \propto \ \Tr[v_1],
\label{eq:cor1}
\end{equation}
thus to ensure the one-point correlation function is zero one has to set trace $v_1$ to be zero as well. An additional gain from this choice of probing operators is the natural elimination of one-point terms while calculating higher correlation functions.

We call some correlation function \textit{trivial},
if during its calculation entire tensor network reduces, leaving only one-point terms and correlation function factorizes.

\subsection{Two point correlation function}

For HaPPY codes, while calculating  
2-point correlation function, one may reduce the entire tensor network, as shown in \cite{ChunJun2022}, and the two points of interest are separated. Thus if we sent $1$-point correlations to zero \eqref{eq:cor1}, the $2$-point correlations resulting from HaPPY codes is trivial and factorize,
\begin{equation}
\label{HaPPy_correlation}
    \la \phi (x_1) \phi (x_2) \ra = \la \phi (x_1) \ra \la \phi (x_2)  \ra  = 0 ,
\end{equation}
which corresponds to the scaling dimension  $\Delta = \infty$. 
Fortunately, in the proposed construction of a hyper-invariant tensor network, there are tensors which does not reduce while calculating $2$-point and higher correlation functions. To characterize them we introduce the following notions.

Let us define a \textit{path} between two boundary subsystems as a subset of tensors (together with its conjugates) from the tensor network, such that in each tensor two and only two of its indices, which are not adjacent, connect it with other tensors in the path or are the boundary indices corresponding to distinguished boundary subsystems. An example of a path is presented in Figure  \ref{fig:network_decudtion} (d). Note that one cannot reduce the path, by contracting it with its conjugate on all indices except the distinct one. It is so because each path tensor has at most two neighboring indices which are neither distinct nor connecting it with other path tensors. 

A special type of path is a \textit{geodesic path}, where for each path's node one edge of the corresponding tile belongs to the same geodesic in AdS space. 
One can think of a geodesic path as a path which in each node takes the same turn, as illustrated in Figure \ref{fig:twopoint} (a).
Notice that the definition of the path is slightly more general than the geodesic path. The example path which is not a geodesic path is presented in \ref{fig:network_decudtion} (d). It turns out that paths are as common as one would think.

\begin{lemma}\label{Lem:unique_path}
For any two boundary indices there exists at most one path connecting them.
\end{lemma}

\begin{proof}
We prove this property by contradiction. Assume that for some boundary indices, there exist at least two paths $t_1$, $t_2$ connecting them. Let $n_a$ and $n_b$ be their last common nodes. Thus we may take two geodesics $l_1$, $l_2$  that include edges of $n_a$, and of the next node along $t_1$ and $t_2$ respectively, that goes between $t_1$ and $t_2$. Since $t_1$ and $t_2$ join at $n_b$, and each path could not cross a geodesic (it could be at most "tangent" to it), the $l_1$ and $l_2$ must intersect at the vertex of some node $n_b'$ closer to $n_a$ than $n_b$. Therefore we constructed two geodesics on a Poincare disc which intersects in two points, which is contrary to hyperbolic geometry. 
\end{proof}

With this information, we can fully characterize the two-point correlation functions.

\begin{thm}
\label{2_point_correlation}
The two-point correlation function between two boundary subsystems is not trivial, unlike for HaPPY codes \eqref{HaPPy_correlation}, only if in the network there exists a path connecting those subsystems. Moreover, in such a case, the two-point correlation function simplifies to

\begin{equation}
\la \phi(x_1) \phi(x_2) \ra=\Tr[ \mathcal{O}] \Tr[\mathbf{\mathcal{W}} (v_1 \otimes v_2)],
\end{equation}

where $\mathcal{W}$ is the operator obtained from the path with its conjugate on all indices except the distinct boundary ones, while  $\mathcal{O}$ is a mapped boundary operator and $v_1$, $v_2$ are traceless probing observables.
\end{thm}

In other words, while calculating two-point correlation functions the entire tensor network reduces, except the tensors belonging to the path between subsystems of interest.

\begin{proof}
If the network has only one, central node the theorem follows from explicit calculations.
Otherwise, if the network has more layers, then between two distinct tensors exists at least one tensor with only two neighbouring connections to the network. While calculating the trace such a tensor can be reduced with its configuration since it has three neighbouring boundary indices. The remainder of this tensor is a trace over its bulk index. 
Moreover, since now the neighbours of this tensor have three neighbouring indices connected with their conjugations, they can be reduced as well. This procedure can be repeated until all possible reductions of the boundary tensors are completed. 

Such reduction on the entire boundary, except the tensors with indices corresponding to distinct boundary subsystems, is always possible if there exists at least one tensor with three boundary indices between those two distinct tensors. It may not be the case only if one of the following scenarios occurs:
\begin{enumerate}
    \item The distinct indices are such that the path between them lies on the boundary of the network. Then none of the tensors on the path would simplify. However, the consecutive moves will reduce the rest of the network.
    \item At least one distinct tensor can be reduced as well, which is possible only if there exists no path connecting distinct boundary subsystems. Then in the consecutive boundary reductions entire network will be reduced and the resulting correlations trivial.
\end{enumerate}
\begin{figure}[h]
    \centering
    \begin{subfigure}[h!]{0.45\textwidth}
        \centering
 \subcaptionbox{}{
            \includegraphics[width=2.5in, height = 2.5in]{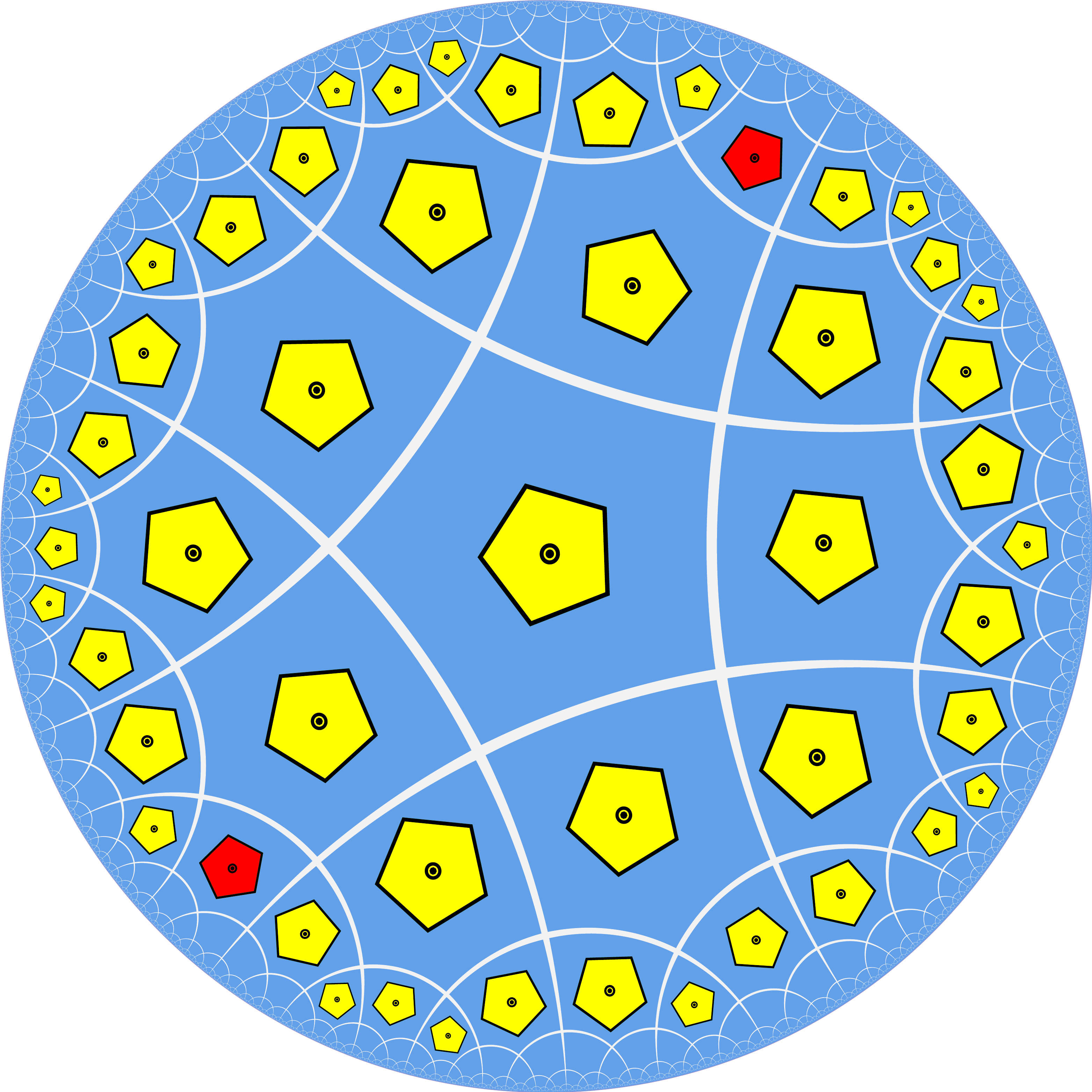}
    }
    \end{subfigure}%
    ~
    \begin{subfigure}[h!]{0.45\textwidth}
        \centering
\subcaptionbox{}{
            \includegraphics[width=2.5in, height = 2.5in]{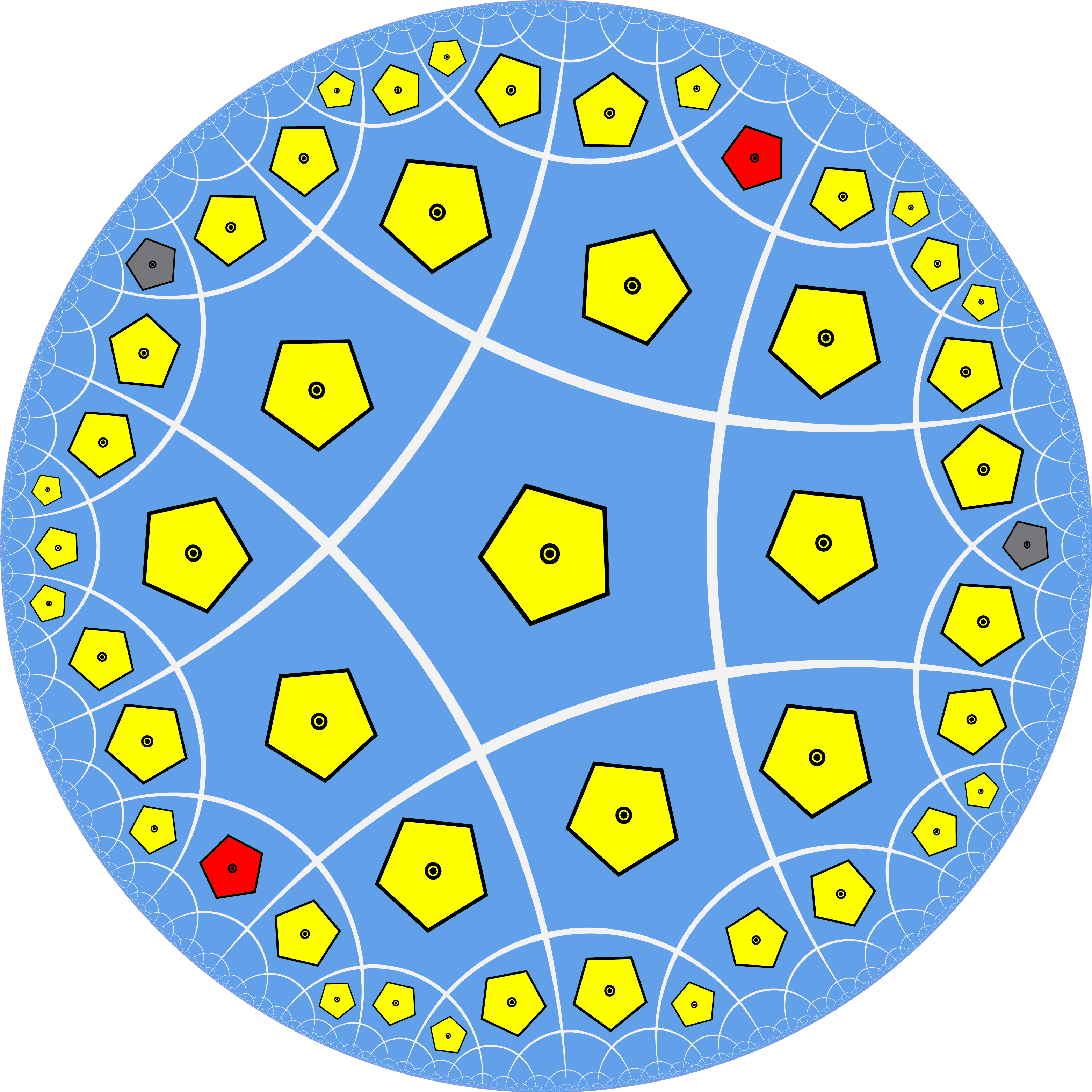}}
        
    \end{subfigure} \\
    \vspace{0.5cm}
    \begin{subfigure}[h!]{0.45\textwidth}
        \centering
        \subcaptionbox{}{
            \includegraphics[width=2.5in, height = 2.5in]{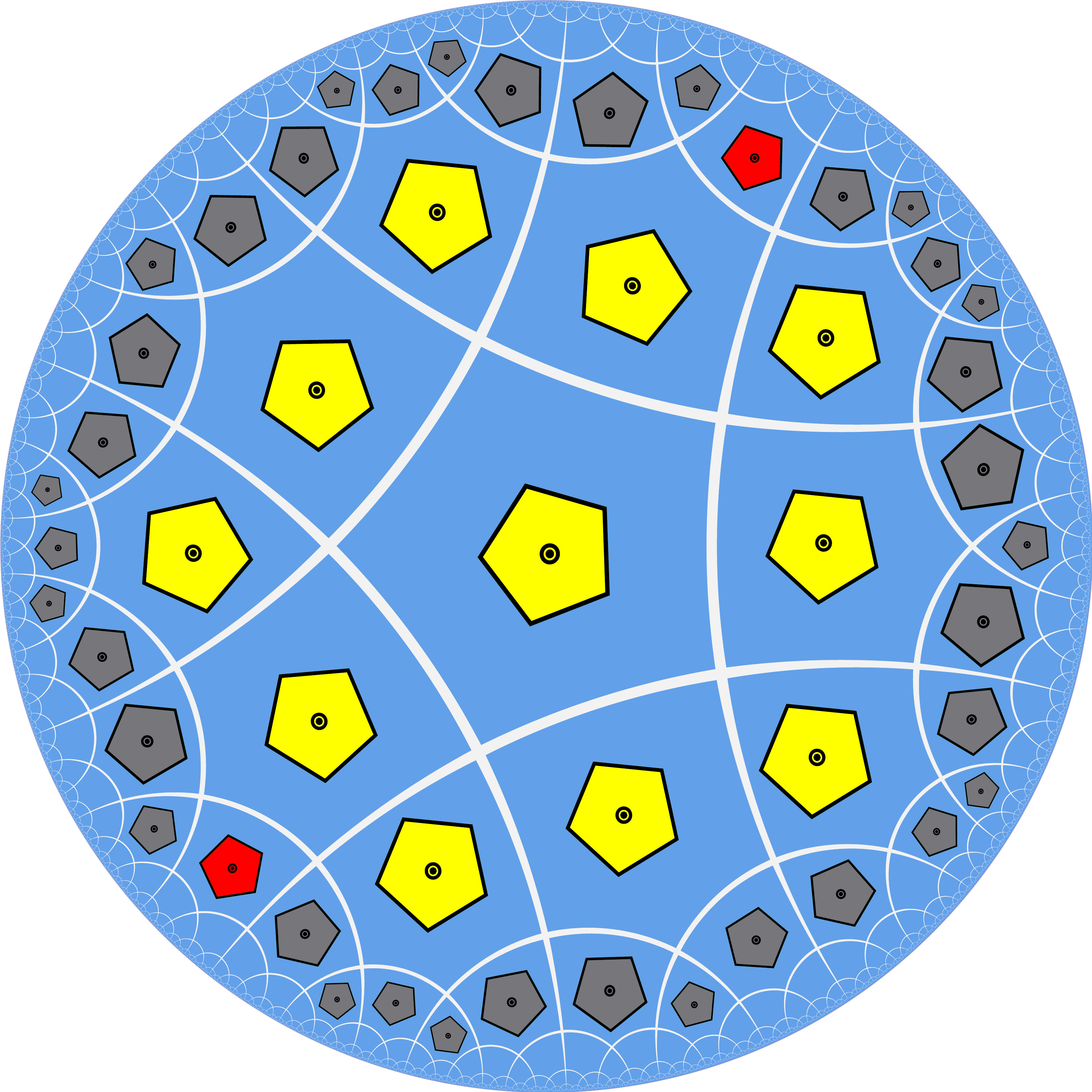}}
    \end{subfigure}%
    ~ 
    \begin{subfigure}[h!]{0.45\textwidth}
        \centering
        \subcaptionbox{}{
            \includegraphics[width=2.5in, height = 2.5in]{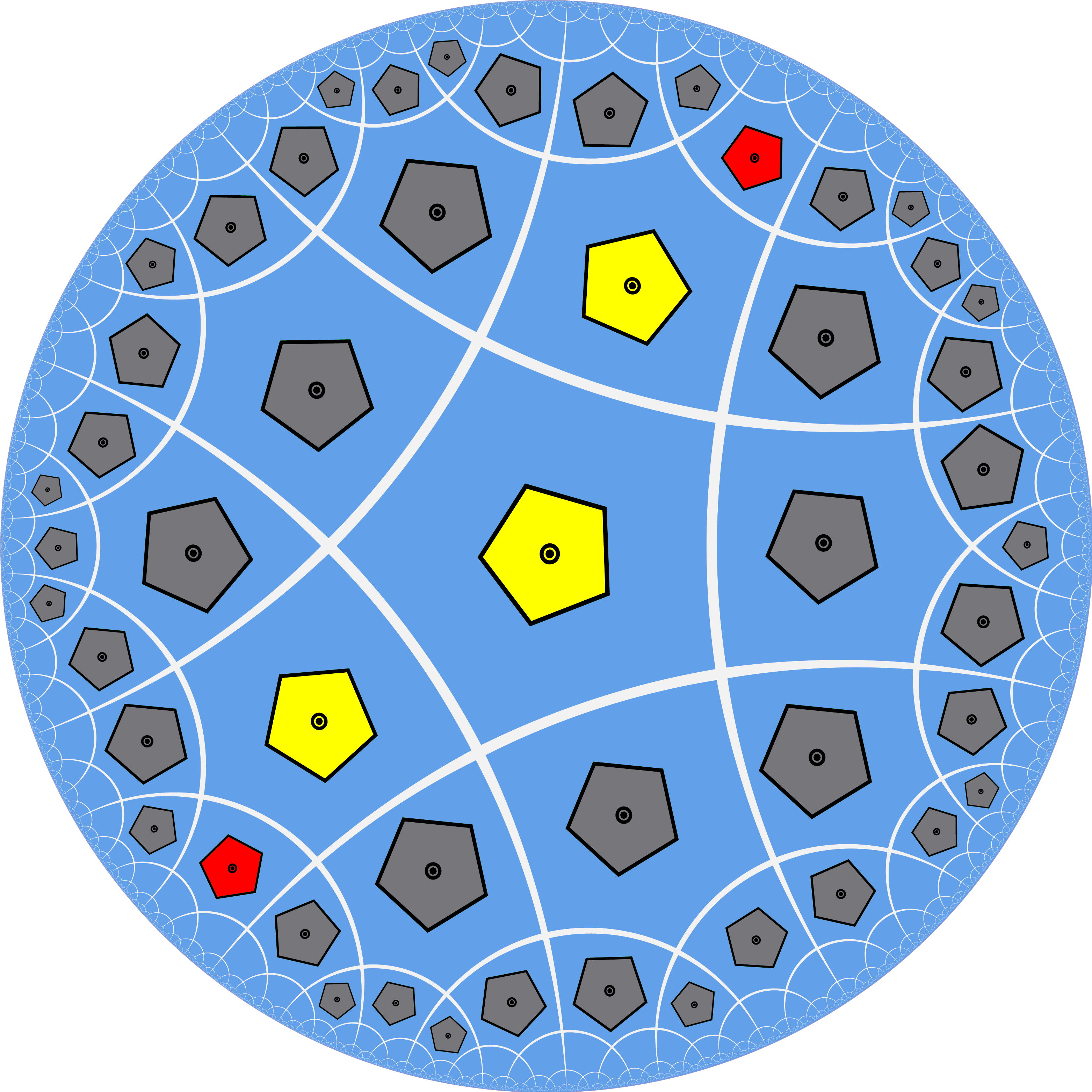}}
    \end{subfigure}%
    
    \caption{
    \textbf{(a)} Calculation of the 2-point correlation function between the 
    points at the boundary with corresponding tensors marked red.
    \textbf{(b)}  two tensors, marked in gray, are reduced, which ensures reduction of other tensors from this layer \textbf{(c)}. Similar procedure is repeated until only the path between distinct indices (red symbols) remains.
     \textbf{(d)}  2-point correlation function can be directly computed.}
    \label{fig:network_decudtion}
\end{figure}

In the consecutive steps, one repeats these steps over the new boundary. This procedure stops only if all tensors are reduced, or none of the remaining tensors can be reduced. The latter implies that the remaining tensors have at least two non-neighboring indices connected with other non-reduced tensors (or these indices are the distinct ones). Thus the boundaries of a set of unreduced tensors must be paths connecting distinct boundary indices. This means, due to the uniqueness of the path, that unreduced tensors are forming a path.
The diagram of the above-discussed actions is presented in Figure \ref{fig:network_decudtion}.

The tensors on the path can also be partially reduced. Since they have $3$ indices connected with their conjugates, the $3$ corresponding outside unitaries simplifies. The next simplification is the reduction of one of the inner unitary. Finally, the perfect tensors can be reduced as well, which results in a trace over its bulk indices. 
Thus we obtained the trace over all bulk indices and eliminated all tensors outside the path which ends the proof.
The reduced path tensor is presented in the Figure \ref{fig:twopoint}.
\end{proof}

Note that not all subsystems are connected by a path within a network, thus not all local subsystems are correlated. It is, in a sense side effect of a discrete network construction, which imposes that not all subsystems can be connected by a geodesic path within such a network. On the other side, we stress once again that the defined notion of the path is more general than the geodesic path.

We can describe the two-point correlation function much more precisely utilizing one more observation: 

\begin{lemma}
\label{Lem:largest_eig}
Each simplified node on the path can be interpreted as a matrix acting on $8$ qudits. Such a matrix have an eigenvector $\sum_{\mathbf{i}}|\mathbf {i}\mathbf {i}\ra$, with $\mathbf {i} =(i_1,i_2,i_3,i_4)$, to the eigenvalue $d^{5} = 32$. Moreover, it is also an eigenvalue with the maximal module of this matrix.
\end{lemma}
\begin{proof}
The first part of the Lemma can be established by a simple calculation, relying on properties of unitarity.
To prove the second part
we use the fact that the spectral radius of any matrix $r(A)$ is not larger than its operator norm,
$r(A) \leq ||A|| =  \max_{\mathbf{x}} ||A |\mathbf{x}\ra||/|||\,|\mathbf{x}\ra||$ \cite{Conway2007}.

Let us consider any vector $|\mathbf{x}\ra$ on which the reduced node acts. By using repetitively three properties
\begin{enumerate}
    \item Unitary matrix does not change the norm of the vector,
    \item Contraction does not enlarge the norm of a tensor,
    \item Tensor product with non-normalized generalized Bell state $\sum_i^{d} |ii\ra$ enlarges the norm $d = 2$ times,
\end{enumerate}
one can show that the norm $||\,|\mathbf{x}\ra||$ can increase at most $d^{5} = 32$ times.
\end{proof}

Note that the leading eigenvector of the node corresponds to the trace of the local boundary observable since its upper and lower indices repeat.
Thus the contribution of the leading term to the two-point correlation function would be a product of local expectation values $\la\phi(x_1)\phi(x_2)\ra \approx \la\phi(x_1)\ra\la\phi(x_2)\ra$. This ensures us that the normalization of the path's nodes should be $\lambda_1 =d^{5} =  32$. Otherwise, this leading therm would either explode or vanish which is not physical for nonzero-trace probing observables.

A similar lemma can be also proven for combinations of nodes on the path, which leads us to the main result of this subsection.

\begin{cor}
\label{cor_2point}
Assume the reduced path node has only one eigenvector to eigenvalue with maximal module $d^{5} = 32$. Then by normalizing the nodes by this factor, in the limit of a large network, one obtains the desired decay behavior of two-point correlation functions \eqref{eq:2point_corr_behaviour}. Moreover, if the bulk operator has normalized trace and the path is a geodesic, the scaling dimension $\Delta$  is given by
\begin{equation}
\Delta = - \log_{\mu} \lambda_2,
\end{equation}
where $\lambda_2$ is  the subleading eigenvalue of the simplified node and $\mu = 2 + \sqrt{3}$ is a scaling factor of a network.
\end{cor}

Consider one step of inflation -- new layer addition, for the tensor network with a large number of layers. From one hand side, the number of tensors grows according to scaling factor $\mu = 2 + \sqrt{3}$ \cite{Central_charges_and_scaling}, resulting in the same growth of distance between two points on the boundary. On the other hand new layer introduces two new nodes, each on opposite ends of the path, which lowers the correlation by a factor equal to the next to the leading eigenvalue squared. We neglect the terms with lower eigenvalues since they decay appropriately faster in a large number of layers limit. Thus we obtain the desired power-low decay of correlations (\ref{eq:2point_corr_behaviour}). If the considered path of interest is a geodesic path, it consists of exactly the same nodes, so we may compare correlations after the inflation step explicitly:

\begin{equation}
\label{Delta_dev}
\begin{aligned}
\la \phi(\mu x_1)\phi(\mu x_2) \ra &=  \frac{1}{\ell_{12}^{2\Delta} \mu^{2 \Delta}} =  \la \phi(x_1)\phi(x_2) \ra \lambda_2^2, \\ 
\frac{1}{\mu^{2\Delta}} &= \lambda_2^2, \\
\Delta &= - \log_\mu \lambda_2 .
\end{aligned}
\end{equation} 

Note that the bulk operator affects the decay only by a scalar factor -- the trace. Therefore,  if the bulk operator has a normalized trace it does not affect the decay.

In the presented model the discussed conformal field theory lies on a cylinder, with the constant time slice being a circle. Thus to be precise one should consider an exact correlation function for the theory mapped onto a cylinder \cite{Ginsparg:1988ui}
\begin{equation}
    \la \phi_{cyl}(x_1) \phi_{cyl}(x_2) \ra =  R^{-2 \Delta} \frac{1}{\left[4 \sin\left(\frac{x_1 - x_2}{R}\right)\right]^\Delta}~,
\end{equation}
where $R$ is the radius of the cylinder. Notice, that in such a consideration, in the inflation step, we must take into account not only the increase of the field's positions but the radius $R$ as well.
Thus the exact inflation step corresponds to
\begin{equation}
\begin{aligned}
\la \phi_{cyl}(x_1) \phi_{cyl}(x_2) \ra \to &\;  {(\mu R)}^{-2 \Delta} \frac{1}{\left[4 \sin\left(\frac{\mu( x_1 - x_2)}{\mu R}\right)\right]^\Delta} =  
{(\mu R)}^{-2 \Delta}  \frac{1}{\left[4 \sin\left(\frac{x_1 - x_2}{R}\right)\right]^\Delta} \\
&  = \frac{1}{\mu^{2 \Delta}} \la \phi_{cyl}(x_1) \phi_{cyl}(x_2) \ra~,
\end{aligned}
\end{equation}
which is consistent with the scaling in Eq.\eqref{Delta_dev}.

In the generic case, the reduced node doesn't need to be a normal matrix, thus the next-to-leading eigenvalue corresponds to a set of Jordan blocks of this matrix rather than just one eigenvector.
Since the general scenario reproduces the same result, we move its technical discussion to the Appendix \myref{App:Jordan}{D}.

\subsection{Three point correlation functions}

Three-point and higher correlation functions differ significantly from two-point one by the much greater impact of the bulk operator on their form. To preset this influence explicitly we start, similarly as previously, by considering all possible reductions one can perform while calculating the correlations.

\begin{thm}
\label{3_point_correlation}
The three-point correlation function between three distinct boundary indices is nontrivial (nonzero) only if there exists a path between two of them and a path leading from the third index to some tensor on the first path. 
While calculating the correlation all tensors will reduce, except the ones on those paths.

Furthermore, there exists no configuration of three boundary indices such that each pair of them is connected by a path.
\end{thm}

We note that the requirements for a non-vanishing three-point correlation function are, once again, a side effect of hyperbolic space discretization. In AdS space any two points can be connected via geodesic, thus the three-point correlation function would not disappear for any points' configuration on the boundary.

\begin{proof}

Let us start by showing that there cannot exist $3$ paths connecting each pair of distinct indices. Lets consider three such paths $t_{1,2}, t_{2,3}, t_{1,3}$, and mark their last  common nodes  $n_1, n_2, n_3$. 

Similarly as in the proof of the Lemma \ref{Lem:unique_path}
one can replace paths $t_{1,3}$, $t_{2,3}$ by two geodesics $l_{1,3}, l_{2,3}$, that contains the edges of $n_1$ and $n_2$, "follows" $t_{2,3}$ and $t_{1,3}$, and are between those paths. Because paths $t_{1,3}$ and $t_{2,3}$ cannot cross geodesics $l_{1,3}$ and $l_{2,3}$, those geodesic must intersecting in the vertex of some node $n_3'$ on the same side as $n_3$ but not further from the path $l_{1,2}$. 
However, in a similar manner, we can also replace the path $t_{1,2}$ with a geodesic $l_{1,2}$, that intersects with $l_{1,3}$ and $l_{2,3}$.

In the used tiling of the AdS space four pentagons meet in each vertex, so the angles between intersecting geodesic are $360^{\text{o}}/4 = 90^{\text{o}}$. Thus using geodesics $l_{1,2}$, $l_{2,3}$ and $l_{1,3}$ we constructed a triangle on the hyperbolic plane with a sum of inner angles larger than $180$ degrees, which gives a contradiction.

The important conclusion from the above property is that there exist no nodes surrounded by paths, and do not belong to any of those.
Therefore using the same procedure, as discussed in the proof of Theorem \ref{2_point_correlation} one can reduce all tensors except ones on the path connecting two distinct indices, and the the path adjoining the last index.
\end{proof}

\begin{figure}[h]
\centering
\begin{subfigure}[h!]{0.5\textwidth}
\centering
\subcaptionbox{}{
            \includegraphics[width=3in]{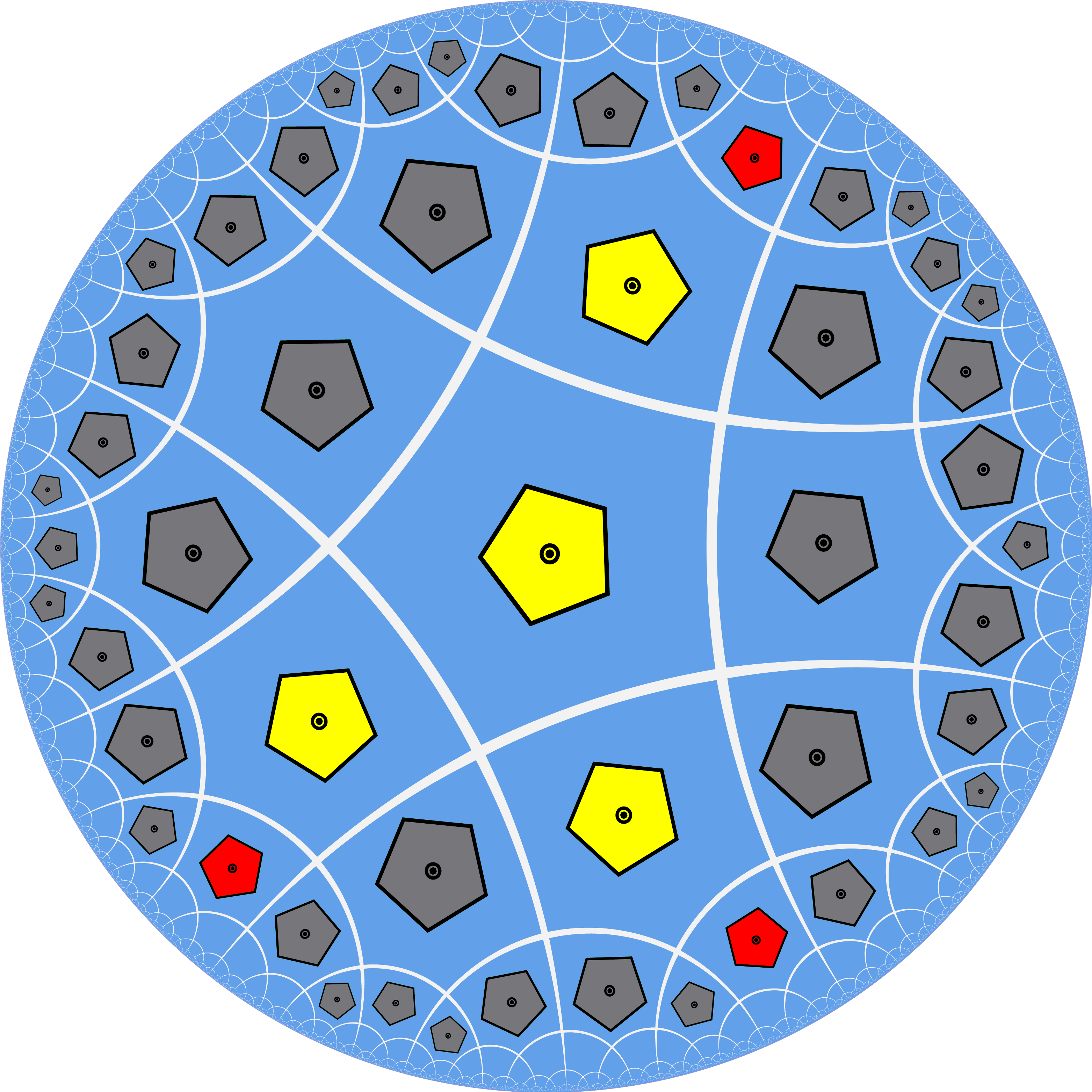}}

\end{subfigure}%
~
\begin{subfigure}[h!]{0.5\textwidth}
\centering
\subcaptionbox{}{
            \includegraphics[width=3in]{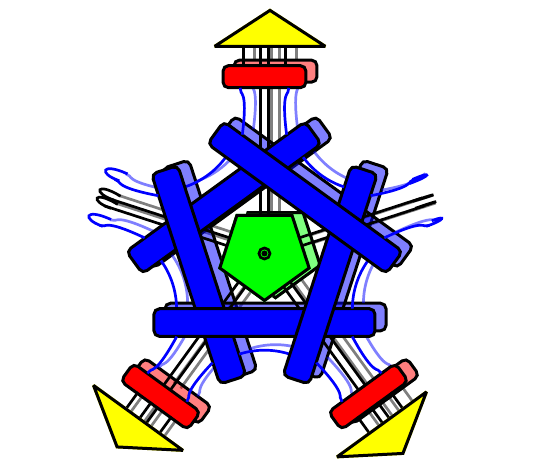}}
\end{subfigure} 
\\
\caption{\textbf{(a)} Exemplary reduction of tensor network for three-point correlation function. Red tensors possess distinguished indices, while gray tensors can be reduced in the same manner as for 2-point correlation.\newline
\textbf{(b)} The central node at the intersections of the paths after all possible reductions. The conjugations of each remaining matrix are denoted as their pale copies. The yellow triangles correspond to projections onto eigenvectors corresponding to the next-to-leading eigenvalues and the bulk indices are yet uncontracted with the bulk operator.}
\label{fig:3-point_corr}
\end{figure}

Note that all bulk indices are traced out during the reduction except the one on the node connecting two paths, in which the perfect tensor is not reduced. Therefore even a local operator, if placed correctly, may affect the correlation function in a non-trivial way.
%
%
\begin{cor}
\label{cor_3point}
With the same assumptions as in Corollary \ref{cor_2point}, the obtained three-point correlation function has the desired form \eqref{3point_corr_behaviour} with the same scaling dimension as 2-point correlation functions \,\,\, $\Delta = -\log_{\mu}\lambda_2$. The proportionality coefficient $C_{123}$ depends
on the reduction of the bulk operator to the node in which paths intersect.
\end{cor}

Consider a leftover tensor network while calculating three-point correlation functions, Fig. \ref{fig:3-point_corr}, as three geodesic paths going from the intersection tensor to distinct indices on the boundary, and the intersection tensors itself. Because each of those paths behaves exactly the same as in the case of two-point correlation functions, each bulk index of the path's tensors is traced out, thus the only remaining bulk index stays in the intersection of paths.

Moreover, if we assume that each geodesic path is long, by the same arguments as in Corollary \ref{cor_2point} the correlations must be described by \eqref{3point_corr_behaviour} with the same scaling dimension. The only new part is the proportionality coefficient $C_{123}$.
One can calculate it by tracing the bulk operator over all except one bulk index and placing it onto the intersection node. Moreover, the connections of the intersection node to the paths may be replaced by the projections to eigenvectors corresponding to next-to-leading eigenvalues of path nodes, as presented in Figure \ref{fig:3-point_corr} (b).

\subsection{Higher correlation functions}

Using the same methods one can extend analysis into other correlation functions. However, the higher the order of the correlation function, the more complex analysis becomes.
New phenomenon that appears at the level of the four-point correlation function, is a set of nodes that can not be reduced,
even though they do not belong to the paths connecting distinct boundary indices. Such an effect does not occur for three-point correlation functions,
since it would lead to the construction of a triangle with three inner angles equal to $\pi/2$ on the hyperbolic plane -- consult 
proof of Theorem \ref{3_point_correlation}.

However, while studying the four-point correlation function one can consider a "rectangle" of nodes of shape $2\times k$, as a segment of two parallel neighboring paths, tangent to a geodesic between them. Such a set of nodes cannot be reduced, for example, if two shorter sides of this "rectangle" are contained in paths that connect distinct boundary indices, see Fig \ref{fig:4point_corr}. 
\begin{figure}[h!]
    \centering
    \includegraphics[width=0.4\linewidth]{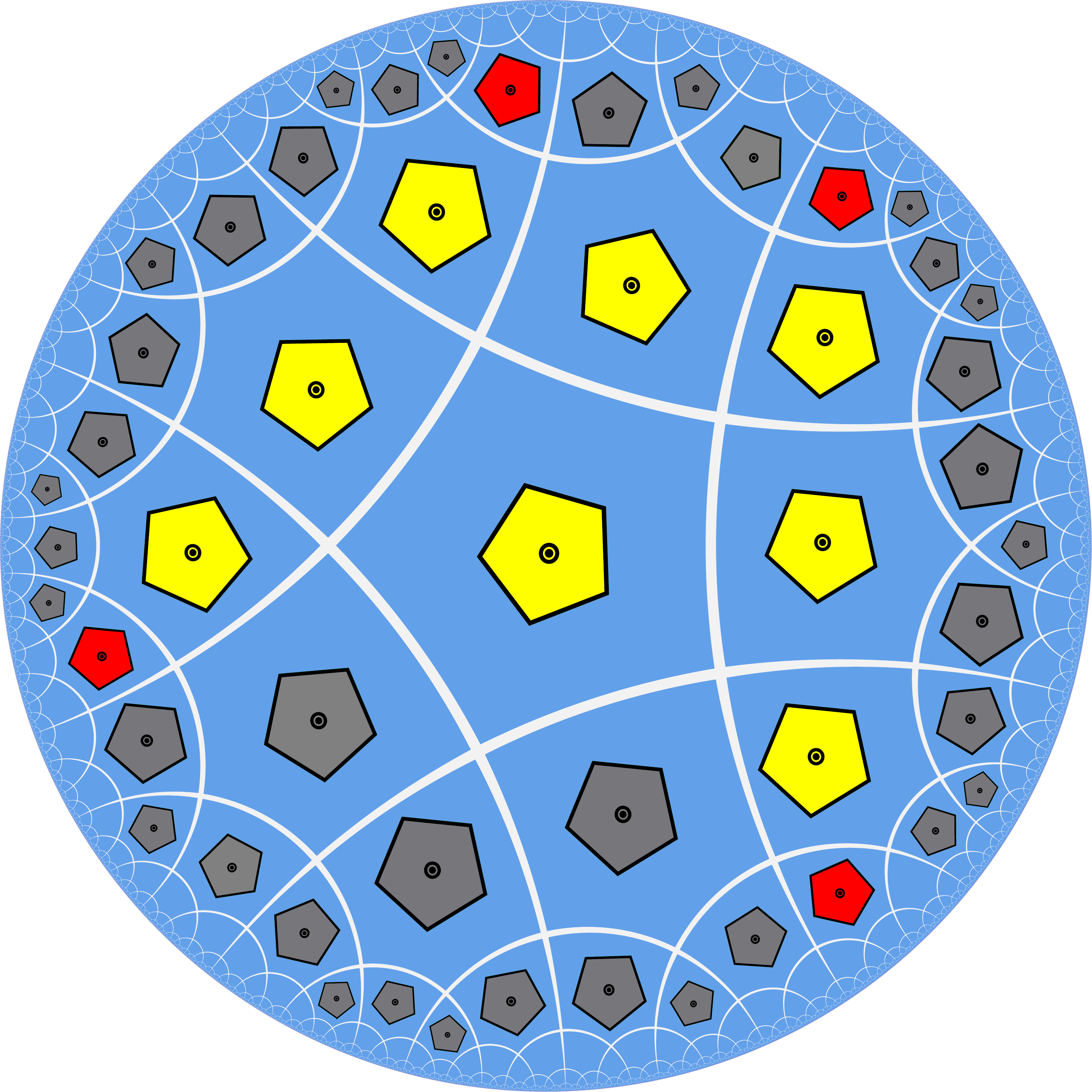}
    \caption{Exemplary reduction of tensor network for correlation function between four boundary indices with corresponding tensors marked red: scenario in which some non-reduced tensors do not belong to paths between distinct boundary indices.}
    \label{fig:4point_corr}
\end{figure}

For higher point correlation functions
such cliques of non-reduced nodes  become larger and more complicated. Similar phenomena were observed in the reductions of hyper-invariant codes \cite{Steinberg2023}. 

Despite those complications, the argument regarding the decay of correlation functions still holds. Adding one more layer to the tensor network will extend each path to a distinct boundary index by one node, resulting, in a large network limit, in the appropriate scaling, analogously as in Corollary \ref{cor_3point}. Furthermore, convoluted formulas for higher point correlation functions, which cannot be reduced into lower order correlations, display a highly interactive nature of speculated field theory simulated at the boundary.

\section{Numerical examples}\label{Sec:numerics}

In this Section, we present a numerical discussion of available scaling dimensions originating from two-point correlation functions between two boundary indices connected via a geodesic path.

As we demonstrated in Corollary \ref{cor_2point} the important information about correlation decay is encompassed in the next to the leading eigenvalue of a reduced node - the building block of the path.
First, we show an exemplary distribution of this eigenvalue and corresponding scaling dimension for two simple 2-parameter families of tensors and then we try to explore the entire range of possible values using random building blocks of the tensor studied.

Simplest, yet nontrivial families of hyper-invariant tensors were obtained by replacing arbitrary $2$ qubit dual (blue) unitary with a one-parameter family
\begin{equation}
\label{blue_u_example}
{U^{(2)}}_{ij,lk} =
\frac{1}{2}
\left(
\begin{array}{cccc}
 1 & 0 & 0 & 0 \\
 0 & 0 & 2 & 0 \\
 0 & 1+e^{i \pi  a} & 0 & 1-e^{i \pi  a } \\
 0 &  1-e^{i \pi  a } & 0 & 1+e^{i \pi  a } \\
\end{array}
\right) = CNOT^a \textit{SWAP},
\end{equation}

with exponent $a\in [0,1]$,
interpolating between the gates \textit{SWAP} and $DCNOT=CNOT\; $\textit{SWAP} -- see Appendix \myref{App:2systems}{C}.

\begin{figure}[h!]
    \centering
    \begin{subfigure}[h!]{0.48\textwidth}
        \centering
 \subcaptionbox{}{
            \includegraphics[width=2.5in]{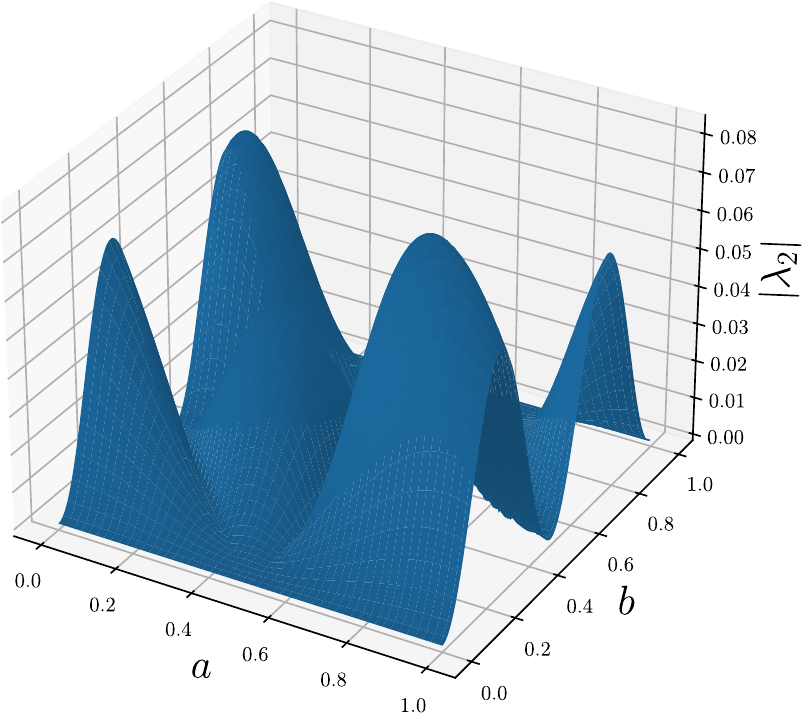}
    }
    \end{subfigure}%
    ~
    \begin{subfigure}[h!]{0.48\textwidth}
        \centering
\subcaptionbox{}{
            \includegraphics[width=2.5in]{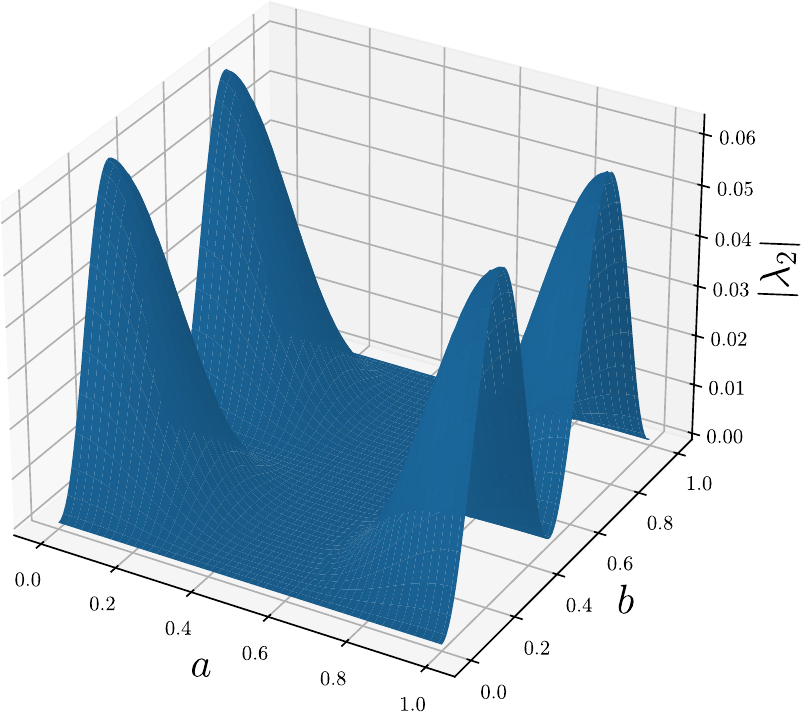}}
        
    \end{subfigure} \\
    \vspace{0.5cm}
    \begin{subfigure}[h!]{0.48\textwidth}
        \centering
        \subcaptionbox{}{
            \includegraphics[width=2.5in]{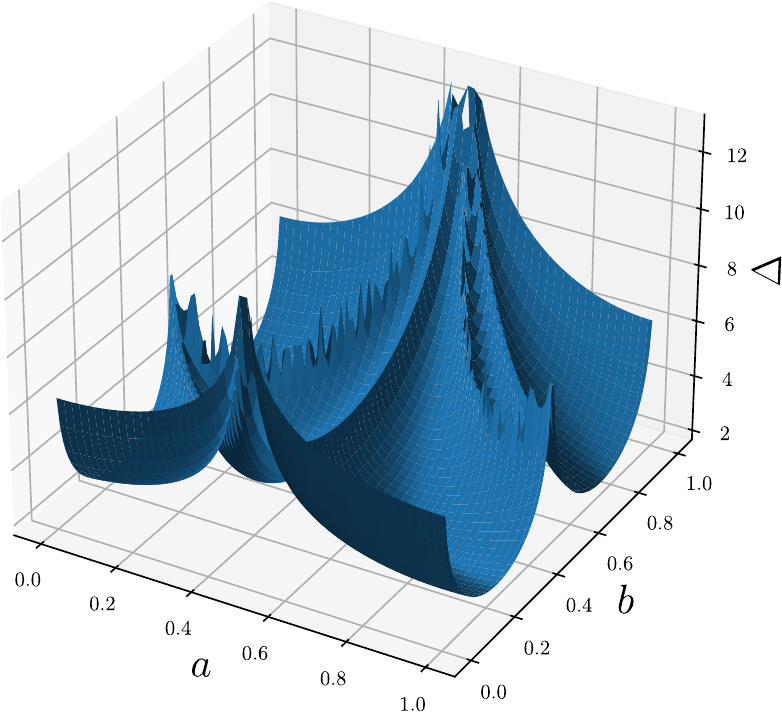}}
    \end{subfigure}%
    ~ 
    \begin{subfigure}[h!]{0.48\textwidth}
        \centering
        \subcaptionbox{}{
                \includegraphics[width=2.5in]{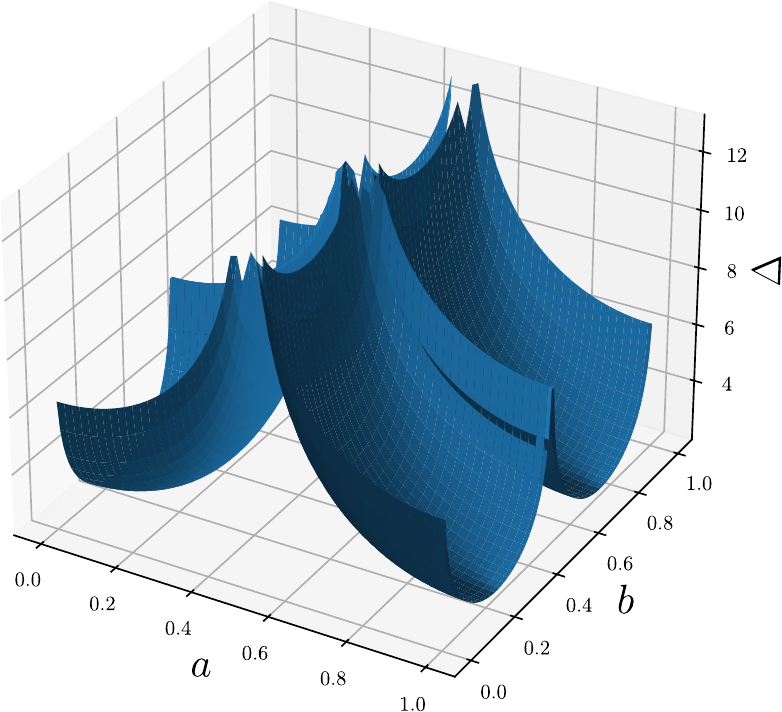}}
    \end{subfigure}%
    
    \caption{Plot of the next-to-leading eigenvalues $\lambda_2$ \textbf{(a,b)} of a reduced node within the geodesic path and corresponding scaling dimensions $\Delta$ \textbf{(c,d)} for nodes constructed according to two simple families discussed in the text \eqref{blue_u_example}, \eqref{red_u_example1} and \eqref{blue_u_example}, \eqref{red_u_example2} correspond to \textbf{(a,c)} and \textbf{(b,d)}. 
    }
    \label{fig:distribution_within_families}
\end{figure}
The freedom to choose example $2$ ququart unitaries is much larger. However, for the sake of simplicity, we restricted ourselves to two cases. First was a product of two $CNOT^{b}$ gates, with power $b \in [0,1]$, interpolating between identity and $CNOT$ (see App. \myref{App:2systems}{C}), each acting on one frame's qubit and one perfect tensor's qubit, with the frame qubit being the control one.
The second family of $2$ ququart unitaries was similar but this time we first on $CNOT$ onto a perfect tensor's qubits and with $CNOT^{b}$ in the opposite direction, and power $b \in [0,1]$, effectively interpolating between $CNOT$ and $DCNOT$ (see App.\myref{App:2systems}{C}).  Thus, those two cases may be summarized as
\begin{equation}
\label{red_u_example1}
U = \left(CNOT_{F\to T}^b \right)^{\otimes 2},  
\end{equation}
\begin{equation}
\label{red_u_example2}
U =  \left(CNOT_{F \to T}CNOT_{T\to F}^b\right)^{\otimes 2} ,
\end{equation}
where the subscript $F\to T$ or $T\to F$ indicates which qubits, the one coming from perfect tensor or hyper-invariant frame, was the control one.
The obtained values of the next-to-leading eigenvalue and corresponding scaling dimensions are presented in Figure \ref{fig:distribution_within_families}.  
For other simple families of similar structures, we found alike behaviour.

To probe the entire range of possible eigenvalues for reduced tensor nodes we considered random dual unitary and $2$ ququart unitary. Both unitaries were sampled with the Haar measure. However, to obtain the required properties of dual unitarity we used a Sinkhorn algorithm, presented example in \cite{36_officers_of_Karol}, using a random unitary sample as a starting point.
The distribution of obtained eigenvalues and corresponding scaling dimensions are presented in Figure \ref{fig:distribution_delta_random}.

\begin{figure}[h!]
    \centering
    \begin{subfigure}[h!]{0.48\textwidth}
        \centering
 \subcaptionbox{}{
            \includegraphics[height=2.2in]{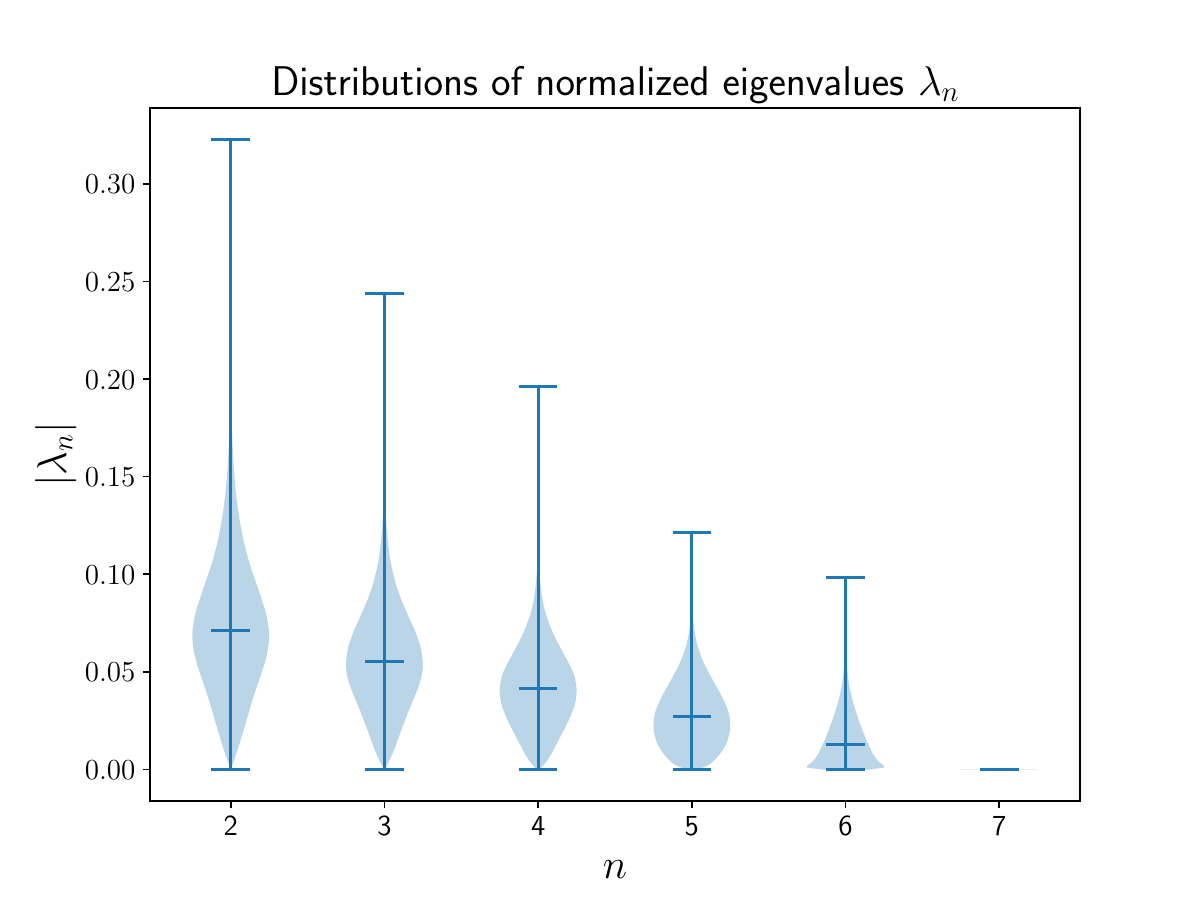}
    }
    \end{subfigure}%
    ~
    \begin{subfigure}[h!]{0.48\textwidth}
        \centering
\subcaptionbox{}{
            \includegraphics[height=2.2in]{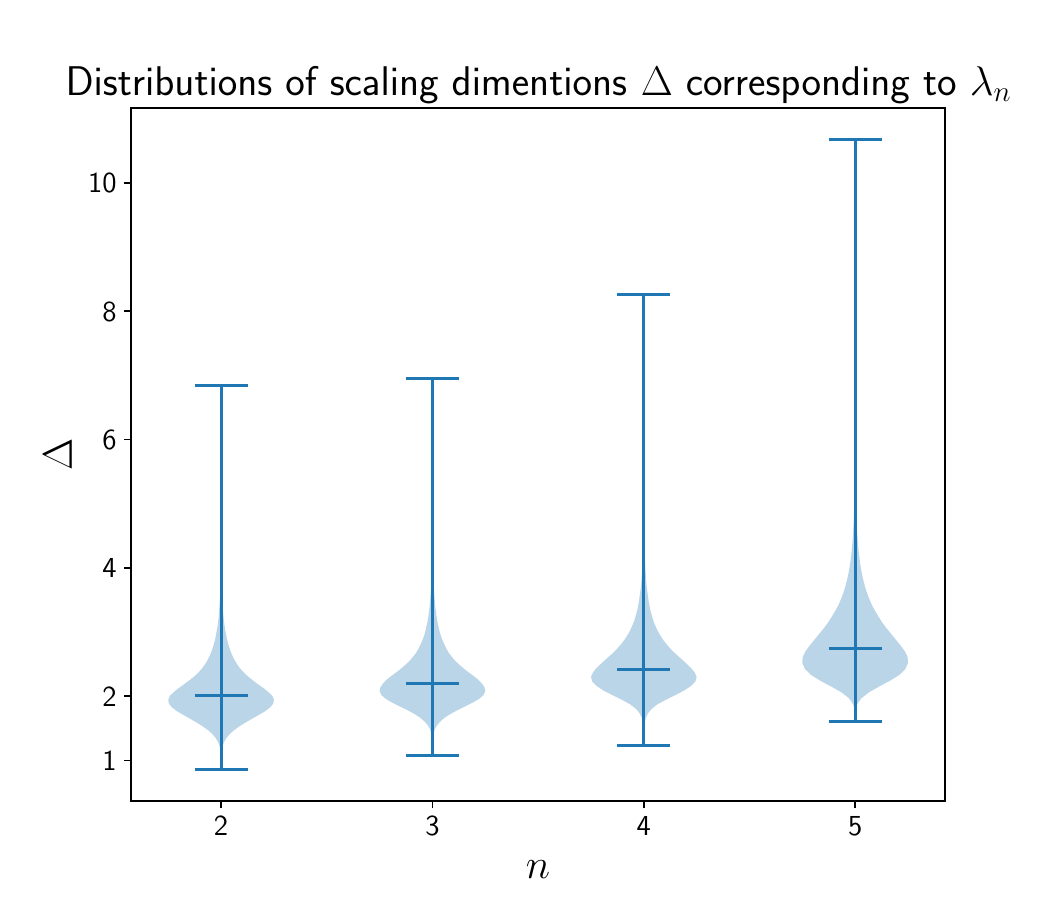}}
    \end{subfigure} 
    \caption{Violin plot of the next-to-leading eigenvalues $\lambda_2,\cdots,\lambda_6$ \textbf{(a)} of a reduced node within the geodesic path and corresponding scaling dimensions $\Delta$ \textbf{(b)} for nodes constructed form random unitary matrices. The number of random samples was equal to $10^6$. The maximal minimal and median values are highlighted.
    }
    \label{fig:distribution_delta_random}
\end{figure}

As we mentioned in the introduction, the HaPPY tensor network does not reproduce the expected form of the two-point correlation functions, unless we interpret them as decaying infinitely fast \eqref{HaPPy_correlation}.
However, if we break some of the HaPPY's symmetry, as we intentionally did in our construction, the scaling dimensions may obtain different nontrivial values. In our numerical study, we find the minimum and maximal values of the scaling dimension to be in the range $\Delta \in (0.858852,6.841722)$. Since the minimal value of the scaling dimension is greater than zero $\Delta > 0$, that guarantees simulated conformal (scalar) field theory to be unitary \cite{Poland:2018epd}. Moreover, since $\Delta = 1$ is within the range on our model, we can simulate the primary fields \(\phi_\varepsilon\) with conformal weights \((\frac{1}{2}, \frac{1}{2})\), which corresponds, for instance,  to the energy \(\varepsilon\) in the Ising model discussed in \cite{Pfiefer2009}. This indicates that novel constructions based on the foundations of the HaPPY model could indeed be tuned to simulate wider ranges of properties for AdS/CFT correspondence.

Similar results were earlier obtained by other approaches -- see  \cite{Steinberg2022}. However, up to our knowledge, none of the derivations concerns the networks simulating bulk-boundary correspondence, led to the expected form of the two-point correlation functions.

\section{Concluding Remarks}

In this work, we constructed and investigated tensor networks inspired by the AdS/CFT correspondence. By combining perfect tensors \cite{HaPPy_code} with hyper-invariant ones \cite{Evenbly2017}, as illustrated in Figure \ref{fig:Fig3}(b), we were able to benefit from the desired properties streaming from both of them. Using perfect tensors we reestablished a simple, yet elegant mapping from the bulk to the boundary Hilbert space, and thanks to the hyper-invariant frame we obtained nontrivial correlations at the boundary.
Our construction is designed to resemble conformal field theory (CFT) on the boundary of a $\{5,4\}$  hyperbolic tessellation of the Poincar{\'e} disk.

Inspired by the field-theoretic perspective, instead of discussing quantum states, we focused on mapping bulk operators to boundary ones by "sandwiching" them between the constructed tensor network and its conjugation. Such an approach 
allows us to achieve 
the desired behavior of two- and three-point correlation functions at the boundary, while considering the image of the bulk operators. 
This indicates that, together with hyperbolic geometry, we have established a bulk-boundary relation for the presented tensor network. The obtained model presents a new perspective for studying hyper-invariant tensor networks through explicit holographic-like mappings, which could be interesting from the viewpoint of quantum error correction,
as recently proposed for
Evenbly codes  \cite{Steinberg:2024ack}.

Due to  flexibility in our construction, we managed to explore a wide range of scaling dimensions $\Delta$ 
within simple examples. This highlights the ability of the model to be tailored for desired behavior, corresponding to certain conformal fields of interest. For instance, $\Delta = 1$ corresponds to energy $\epsilon$ in the Ising model at the critical point \cite{Pfiefer2009}.

The natural direction of further study would be to generalize the presented tensor network for different tessellations of Poincar\'e disc. Our methodology advocated here for perfect tensors with $6$ indices, can be applied to any other perfect tensor, with order greater or equal $6$ and "rotational symmetry".
This approach may directly connect the curvature of AdS space, represented by the different tilings, with the properties of conformal field theory, expressed in correlation functions.
The question, 
whether the proposed approach enables one to obtain a discrete counterpart of the Virasoro algebra remains open.
Such a task, already achieved for quantum CFT \cite{Jahn2022},
allows one to recover the energy-momentum tensor of boundary conformal field theory. Addressing this problem can
provide insights into the gravitational aspects of a compatible bulk theory, thereby deepening the connection between holography and tensor networks. This direction of research could substantially advance comprehension of the interplay between quantum gravity and quantum information theory.

\medskip
\noindent
{\it Acknowlegments}

It is a pleasure to thank Mario Flory, Felix Huber, Romuald Janik, Gerard Munn{\'e}, Zahra Raissi and Marcin Rudzi{\'n}ski for their fruitful remarks and suggestions.
This work was supported by the National Science Centre, Poland, under the contract number 2021/03/Y/ST2/00193
within the DQuant QuantERA II project
that has received funding from the European Union’s Horizon 2020 research and innovation programme under Grant Agreement No 101017733.

\bigskip

\appendix

\section{Appendix A. Building blocks of the construction proposed}\label{App:Tensor_node}

In this Appendix, we present a complete and explicit formula for the construction of the proposed hyper-invariant tensor, following the steps presented in the main part of the work.

The first ingredient of our recipe is a perfect tensor  $T$ of rank $6$ with cyclic symmetry on all except first index $T_{i,jklnm} = T_{i,njklm}$. Although we didn't manage to obtain such a form for the smallest possible local dimension $2$, the construction of a clear example can be provided already for the perfect tensor of the local dimension $D = 4$.
As perfect tensor we choose a construction of three orthogonal Latin cubes \cite{Goyeneche:2015fda} presented in Figure \ref{hyper_orto_fig}. Latin cube is an example of a combinatorial design, which can be represented as a cube $L_{a,b,c}$ with discrete values such that each hyper-row of the cube is filled with numbers $1,\cdots, D$ without repetition. Orthogonality of Latin cubes can be expressed as an orthogonality between their consecutive layers treated as Latin squares \cite{bistron2023genuinely}.

Using three orthogonal Latin cubes $L_{a,b,c}^{(1)}, L_{a,b,c}^{(2)}, L_{a,b,c}^{(3)},$ one can construct perfect tensor of minimal support  in the following way

\begin{equation}
\label{multi_u_Latin}
T_{ijklmn} = \left\{
\begin{matrix}
&1& \;\text{ if }\; i = L_{lmn}^{(1)},\;  j = L_{lmn}^{(2)}\;,  k = L_{lmn}^{(3)} \\
&0& \text{ otherwise } \\
\end{matrix}
\right. ~,
\end{equation}
with each index in this example going from $1$ to $4$.

\begin{figure}[ht]
    \begin{center}
        \includegraphics[width=4in]{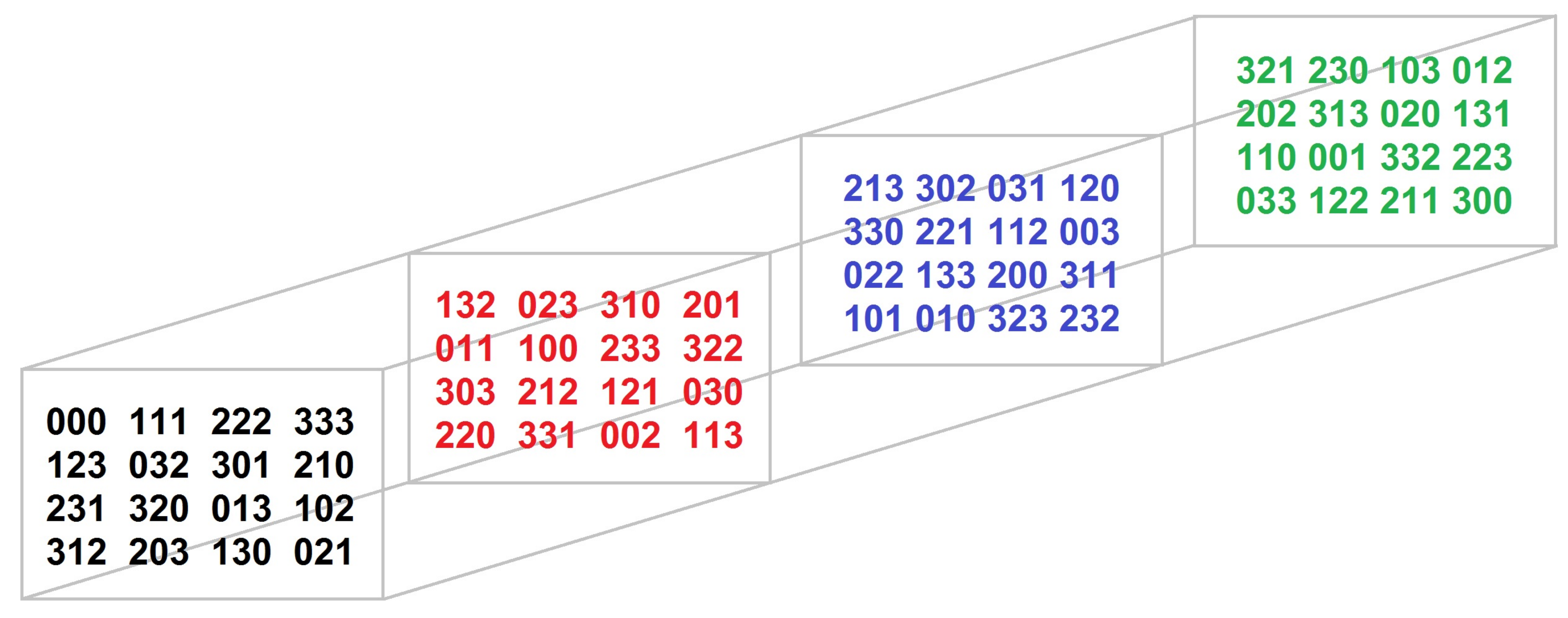}
    \end{center}
    \caption{\label{hyper_orto_fig}
\small{Three orthogonal Latin cubes of dimension $D = 4$ corresponding to 
an AME state of six ququarts and
the perfect tensor with six indices running from one to four
\eqref{multi_u_Latin}. 
Figure borrowed from~\cite{Goyeneche:2015fda} }}.
\end{figure}

The perfect tensor constructed from three Latin cubes presented in Figure \ref{hyper_orto_fig} doesn't have the desired symmetry yet. However, it can be obtained if we suitable permute the values of its four last indices $\tilde{T}_{ijklmn} = T_{ij \pi(k)\rho(l)\sigma(m)\tau(n)}$, with the permutations
\begin{equation*}
\pi = \rho = \tau =  (1, 3, 4, 2) ~,~~~ \sigma = \id = (1, 2, 3,4)~,
\end{equation*}
which can be also interpreted as appropriate permutations of layers in Latin cubes and indices the third one of them.

The next step of our recipe is the construction of a hyper-invariant frame $F$. Similarity as \cite{Evenbly2017}, we build then with bipartite double unitary matrices, which satisfy not one, but two orthogonality relations
\begin{equation}
\label{double_u_explicit}
{U^{(2)}}_{ij,kl} {\overline{U}^{(2)}}_{ab,kl} =\delta_{ia}\delta_{jb} ~,~ {U^{(2)}}_{ij,kl} {\overline{U}^{(2)}}_{aj,bl} =\delta_{ia}\delta_{kb},
\end{equation}
where each index corresponds to one subsystem of ququart -- qubit, thus the unitary has two "incoming" and two "outgoing" indices.
For further discussion and examples of such kinds of matrices see Appendix \myref{App:2systems}{C}.

Using such kind of matrices we can create a hyper-invariant frame $F$ as presented in Figure \ref{fig:Fig1}(b). The explicit formula for the frame is given by
\begin{equation}
\label{frame_construction}
F_{abcdefghij} = \sum_{klmno} {U^{(2)}}_{ibnk} {U^{(2)}}_{adol} {U^{(2)}}_{cfkm} {U^{(2)}}_{ehln} {U^{(2)}}_{gjmo}~.
\end{equation}

Finally, to construct the entire node $N$ we connect bulk indices of perfect tensor and hyper-invariant frame by unitary matrices. To do so, we group frame indices in pairs (ex.$ab\to v$)  $F_{abcdefghij}\to F_{vwxyz}$ and then act with the same unitary $U$ on each pair of frame indices and index of perfect tensor as presented in the Figure \ref{fig:Fig3}(b).
Note that in the discussed example local dimension of a perfect tensor $T$ is a square of the local dimension of a dual unitary gate: $D = d^2$.
Thus the proposed hyper-invariant tensor node $N$ has a form
\begin{equation*}
N_{i,jvkwlxnymz} = \sum_{j'v'k'w'l'x'n'y'm'z'} U_{jv,j'v'} U_{kw,k'w'} U_{lx,l'x'} U_{ny,n'y'} U_{mz,m'z'} \tilde{T}_{i'j'k'l'm'n'} F_{v'w'x'y'z'}~.
\end{equation*}
In the further consideration we group boundary indexes outgoing in each direction into one (ex.$jv\to\alpha$),  $N_{i,jvkwlxnymz} \to N_{i,\alpha\beta\gamma\delta\epsilon}$.

\section{Appendix B. On hyper-invariant tensor networks}\label{App:Hyperinv}

For the convenience of the readers in this Appendix we present
the definition of hyperinvariant tensor network and the Evenbly code, consistent with the notation from the rest of the work, following \cite{Evenbly2017, Steinberg:2024ack}.

\begin{defn}
Consider a $\{p,q\}$ tessellation of Poincar\'e disc, with $p$-gon tiling, such that $q$ of them meet in each vertex, with $p >3$.  A \textbf{hyper-invariant} tensor network is constructed from two types of tensors, $A_{i_1 \cdots i_p}$ placed in the tiles and $B_{j_1 j_2}$ connecting them on the edges, subject to the following criteria:

\begin{enumerate}
    \item $A$ is symmetric under cyclic index permutations, $A_{i_1 i_2 \cdots i_p} = A_{i_p i_1 \cdots i_{p-1}}$,
    \item $B$ is a symmetric unitary matrix, $B_{j_1 j_2} = B_{j_2 j_1}$ and $B B^\dagger = \id$,
    \item $A$ and $B$ satisfy following the isometry constraints:
\begin{equation*}
\begin{aligned}
& \sum_{i_2\cdots i_p}A_{i_1i_2\cdots i_p} B_{i_2 j_2}\cdots B_{i_p j_p} = V_{i_1; \;j_2 \cdots j_p}~, \\
& \sum_{\substack{i_2 \cdots i_p \\ i_2' \cdots i_p'}}  A_{i_1 i_2 \cdots i_p} B_{i_p i_2'} A_{i_1' i_2' \cdots i_p'} B_{i_2 j_2} \cdots B_{i_{p-1} j_{p-1}} B_{i_3' j_3'} \cdots B_{i_{p}' j_p'} = W_{i_1 i_1'; \; j_2\cdots j_{p-1} j_3' \cdots j_{p}'}~,
\end{aligned}
\end{equation*}
    with $V$, $W$ being some isometries from the product of spaces labelled with $i$ indices into the product of spaces labelled with $j$ indices.
\end{enumerate}
\end{defn}

An Evenbly code, called also hyper-invariant code \cite{Steinberg2023}, is a modification of the hyper-invariant tensor network, in which tensors $A$ have an additional index $i_0$ used to encode information from the bulk. The following definition is a slightly relaxed version of the one form \cite{Steinberg:2024ack}. 

\begin{defn}
Consider a $\{p,q\}$ tessellation of Poincar\'e disc, with $p$-gon tiling, such that $q$ of them meet in each vertex, with $p >3$. An \textbf{Evenbly code} is a tensor network constructed from two types of tensors, $A_{i_0, i_1\cdots i_p}$ placed in the tiles and $B_{j_1 j_2}$ connecting them on the edges, subject to the following criteria:

\begin{enumerate}
    \item[1.'] $A$ is symmetric under cyclic index permutations of boundary indices, $A_{i_0,i_1 i_2 \cdots i_p} = A_{i_0,i_p i_1 \cdots i_{p-1}}$,
    \item[2.'] $B$ is a symmetric unitary matrix, $B_{j_1 j_2} = B_{j_2 j_1}$ and $B B^\dagger = \id$,
    \item[3.'] $A$ and $B$ satisfy following isometry constrains:
\begin{equation*}
\begin{aligned}
& \sum_{i_2\cdots i_p}A_{i_0,i_1i_2\cdots i_p} B_{i_2 j_2}\cdots B_{i_p j_p} = V_{i_0,i_1; \;j_2 \cdots j_p}'~, \\
& \sum_{\substack{i_2 \cdots i_p \\ i_2' \cdots i_p'}}  A_{i_0,i_1 i_2 \cdots i_p} B_{i_p i_2'} A_{i_0,i_1' i_2' \cdots i_p'} B_{i_2 j_2} \cdots B_{i_{p-1} j_{p-1}} B_{i_3' j_3'} \cdots B_{i_{p}' j_p'} = W_{i_0 i_0',i_1 i_1'; \; j_2\cdots j_{p-1} j_3' \cdots j_{p}'}'~,
\end{aligned}
\end{equation*}
    with $V'$, $W'$ being some isometries from the product of spaces labelled with $i$ indices into the product of spaces labelled with $j$ indices,
    \item[4.] $A_{i_0,i_1 i_2 \cdots i_p}$ defines isometry from the space labeled by $i_0$ index into the product of spaces labelled with $i$ indices.
\end{enumerate}
\end{defn}

\section{Appendix C. On bipartite unitary matrices}\label{App:2systems}

Dual unitary matrices \cite{BKP19,BRRL24}
are a building stone of the proposed construction \eqref{frame_construction}.
In this Appendix, we discuss the bipartite unitary matrices of order $d^2$,
demonstrating particular features of
the subset of dual unitary matrices.
Most of the discussion concerns two-qubit gates, $d^2 = 4$,
however, some of the below-mentioned properties hold also for arbitrary local dimension $d$.

Any two-qubit unitary matrix $U_{AB}\in U(4)$
can be conveniently represented in
its Cartan form \cite{2_qubit_gate,HVC01},
\begin{equation}
U_{AB} = (u_A \otimes u_B) U_{int} (v_A \otimes v_B)~,
\end{equation}
where $u_A \otimes u_B$
and $v_A \otimes v_B$ represent local unitaries, while the matrix $U_{int}$ describes the interaction between both subsystems,
\begin{equation*}
U_{int} = \exp{\left\{\sum_i \alpha_i \sigma_i \otimes\sigma_i\right\}}~.
\end{equation*}
Here $\sigma_1, \sigma_2$ and $\sigma_3$,  
represent Pauli matrices,
while the phases  $\alpha_i \in [0, 2\pi]$ characterizing interaction strength.
The choice of one-qubit gates does not affect the orthogonality relations \eqref{double_u_explicit}, thus we may focus directly on the interaction part. 
Moreover, to avoid over-parametrization, one can restrict the values of parameters $\pi/4 \geq \alpha_1\geq \alpha_2\geq \alpha_3\geq 0$, effectively creating a simplex, which forms half of the
Weyl chamber \cite{ZVSW03,Weyl_chamber}.

Its vertices are:
\begin{equation}
Id =
\left(
\begin{array}{cccc}
 1 & 0 & 0 & 0 \\
 0 & 1 & 0 & 0 \\
 0 & 0 & 1 & 0 \\
 0 & 0 & 0 & 1 \\
\end{array}
\right), 
CNOT =
\left(
\begin{array}{cccc}
 1 & 0 & 0 & 0 \\
 0 & 1 & 0 & 0 \\
 0 & 0 & 0 & 1 \\
 0 & 0 & 1 & 0 \\
\end{array}
\right),
DCNOT =
\left(
\begin{array}{cccc}
 1 & 0 & 0 & 0 \\
 0 & 0 & 0 & 1 \\
 0 & 1 & 0 & 0 \\
 0 & 0 & 1 & 0 \\
\end{array}
\right), 
\textit{SWAP} =
\left(
\begin{array}{cccc}
 1 & 0 & 0 & 0 \\
 0 & 0 & 1 & 0 \\
 0 & 1 & 0 & 0 \\
 0 & 0 & 0 & 1 \\
\end{array}
\right), 
\nonumber
\end{equation}
corresponding to the following cases
\begin{equation*}
\begin{aligned}
Id &\cong (\alpha_1 = \alpha_2 = \alpha_3 = 0), \\
CNOT &\cong (\alpha_1  = \pi/4, \alpha_2 = \alpha_3 = 0), \\
DCNOT &\cong (\alpha_1  = \alpha_2 = \pi/4, \alpha_3 = 0), \\
\textit{SWAP} &\cong (\alpha_1 = \alpha_2 = \alpha_3 = \pi/4).
\end{aligned}
\end{equation*}

It turns out \cite{jonnadula2017impact} that the entire family of  qubit dual unitary matrices \eqref{double_u_explicit} can be characterized by a single parameter $\alpha_3 \in [0,\pi/4]$, with $\alpha_1 = \alpha_2 = \pi/4$. Thus they lie on the edge interpolating between $DCNOT$ and \textit{SWAP} gates. An alternative parametrization of this family, up to local unitaries, is presented in \eqref{blue_u_example}.

The Weyl chamber is more than just a nice way to visualize the set of bipartite unitary gates, it captures its geometry as well. 
For instance, the work \cite{Dual_u_max_dist} proves that the dual unitary gates in any dimension maximize the distance from the product of local gates, represented as $Id$ vertex.

Another worth-mentioning property of dual unitary matrices is their operator Schmidt decomposition into products of matrices acting on both subsystems separately \cite{PhysRevA.67.052301}. Unitarity condition for the matrix with reshuffled indices (\ref{double_u_explicit}, right) implies \cite{Dual_u_max_dist} that all the Schmidt coefficients are equal. This fact assures the maximal entanglement entropy in the space of all unitary matrices in $U(d^2)$, 
which can be interpreted that the dual unitary form the set of 'maximally non-local' bipartite gates.

\section{Appendix D. Decay of two-point correlation function}\label{App:Jordan}

In this Appendix, we generalize the discussion of Corollary  \ref{cor_2point} to non-normal matrices.
To discuss the high powers of such matrices it is convenient to represent them in the Jordan form $M = \beta  J \beta^{-1}$, where $\beta$ is an invertible matrix, and $J$ is a block--diagonal matrix, with each block having repeated eigenvalue of $M$ on the diagonal, and $1$ above the diagonal.
Similarly, as in the case of normal matrices, we focus only on the next to leading eigenvalue, since the terms involving powers of other eigenvalues decay exponentially faster.

Raising the matrix to a power $M^n = \beta  J^n \beta^{-1}$ result in raising each block of $J$ independently. Thus we may focus on one block:  
\begin{equation*}
\begin{bmatrix}
 \lambda_2 & 1 & 0 & \cdots & 0 \\
 0 & \lambda_2 & 1 & \cdots & 0 \\
 0 & 0 & \lambda_2 & \cdots & 0 \\ 
 \vdots & \vdots & \vdots & \ddots & \vdots \\
 0 & 0 & 0 & \cdots & \lambda_2
\end{bmatrix}^n
=\begin{bmatrix}
 \lambda_2^n & \tbinom{n}{1}\lambda_2^{n-1} & \tbinom{n}{2}\lambda_2^{n-2} & \cdots   & \tbinom{n}{k-1}\lambda_2^{n-(k-1)} \\
 0  & \lambda_2^n & \tbinom{n}{1}\lambda_2^{n-1} & \cdots   & \tbinom{n}{k-2}\lambda_2^{n-(k-2)} \\
 0  & 0  & \lambda_2^n & \cdots   & \tbinom{n}{k-3}\lambda_2^{n-(k-3)} \\ 
 \vdots  & \vdots  & \vdots  & \ddots  & \vdots \\
 0  & 0  & 0  & 0   & \lambda_2^n
\end{bmatrix}
\end{equation*}
where $k < d^8 = 256$ is the size of the block.

Notice, that in the limit of large $n$, the dominant element in the block is the upper right corner $\tbinom{n}{k-1}\lambda_2^{n-(k-1)}$.  Indeed, calculating the ratio between it and any other element $\tbinom{n}{l-1}\lambda_2^{n-(l-1)}$, we obtain
\begin{equation}
\label{app:in_block_ratio}
\lim_{n\to \infty}   \frac{\tbinom{n}{l-1}\lambda_2^{n-(l-1)}}{\tbinom{n}{k-1}\lambda_2^{n-(k-1)}} = \lim_{n\to\infty} \frac{(n-(k-1))!}{(n-(l-1))!}\frac{(k-1)!}{(l-1)!} \lambda_2^{k-l} = 0,
\end{equation}
where in the last step we used the fact that $l < k$. Thus in the large $n$ limit, it is sufficient to consider only this matrix $J$ element.

Each inflation step corresponds, for the geodesic path, to enlarging the power of the matrix $M$ by $2$, thus the element of interest is enlarged approximately
\begin{equation}
\label{app:coef_ratio}
\lim_{n\to \infty}   \frac{\tbinom{n+2}{k-1}\lambda_2^{n+2-(k-1)}}{\tbinom{n}{k-1}\lambda_2^{n-(k-1)}} = \lim_{n\to\infty} \frac{(n+2)(n+1)}{(n+3-k)(n+1-k)} \lambda_2^{2} = \lambda_2^{2}.
\end{equation}
The above approximations, valid in the large $n$ limit, correspond to approximating the powers of matrix $M$ by
\begin{equation}
\label{large_n_M}
M^n \approx \frac{1}{D}(\sum_{\mathbf{i}}|\mathbf{i}\mathbf{i}\ra)(\sum_{\mathbf{i}}\la\mathbf{i}\mathbf{i}|) + \tbinom{n}{k-1}\lambda_2^{n-(k-1)} |\mathbf{v_R}\ra\la\mathbf{v_L}|
,\end{equation}
where $\la\mathbf{v_L}|$ and $|\mathbf{v_R}\ra$ are appropriate row and column of $\beta^{-1}$ and $\beta$ respectively, and can be interpreted as generalized eigenvectors of $M$ to the eigenvalue $\lambda_2$. Therefore, by the property \eqref{app:coef_ratio}, we effectively recreated the situation discussed in Corollary \ref{cor_2point}.

In the case of multiple eigenvectors, or Jordan blocks corresponding to the eigenvalue $\lambda_2$, the result also holds because the decay of the terms corresponding to $\lambda_2$ is alike.

\section{Appendix E. Numerical calculation of three-point correlation functions}

The numerical calculation of multiplicative constant $C_{123}$, in the large network limit, was somewhat more complicated than the computation of scaling dimensions.
The scale of the correlation function is dependent on the overlap between the arbitrary probing operators and the tensor network. This introduces an artificial unknown in the normalization of three-point correlation functions, Nevertheless, it is valuable to realize how the bulk affects the correlation function. 
Thus, in the following, for clarity, we set the above-mentioned term to unity.

The main difficulty in  calculations of 
$C_{123}$ lies in the necessity to obtain the appropriate generalized eigenvector corresponding to the subleading $\lambda_2$. In the generic scenario, there is no guarantee, that matrix $M$, corresponding to the simplified path's node (see Fig. \ref{fig:twopoint} (b)), is normal, or has only one generalized eigenvector to $\lambda_2$.
Furthermore, the Jordan decomposition is famously numerically unstable by its very nature, thus we had to perform some estimations to obtain an approximated form of the desired vector.

For each choice of unitary matrices building hyper-invariant tensor, we started by calculating the matrix $M$, corresponding to the simplified path's node, and subtracting the trivial term corresponding to the largest eigenvalue, discussed in the Lemma \ref{Lem:largest_eig}. Next, we divided the leftover by the next to the leading eigenvalue, to mitigate numerical errors, and raised it to a large power. According to the discussion in the previous section \eqref{large_n_M}, the resulting matrix should have a form
\begin{equation}
\label{app:monster}
N := \frac{\left(M - \frac{1}{D}(\sum_{\mathbf{i}}|\mathbf{i}\mathbf{i}\ra)(\sum_{\mathbf{i}}\la\mathbf{i}\mathbf{i}|)\right)^n}{\lambda_2^n} \approx \tbinom{n}{k-1} |\mathbf{v_R}\ra\la\mathbf{v_L}| + O\left(n^{-1}\right)~,
\end{equation}
where the reminder is the worst-case correction from other elements in the same Jordan block, in accordance with \eqref{app:in_block_ratio}.

With this form, we could perform singular value decomposition (SVD) of \eqref{app:monster},\newline $N = U \Sigma V$, with the row of $V$ and column of $U$ corresponding to the largest singular value approximating $\la\mathbf{v_L}|$ and $|\mathbf{v_R}\ra$.
In practice, we chose $n = 20$ and dismissed the cases in which the ratio between the leading and next to leading singular value of $N$ were smaller than $100$ to ensure the accuracy of the approximation. 

\begin{figure}[h!]
\centering

\subcaptionbox{\,\,\,\,\, \hspace{5.5cm} (b)}{
            \includegraphics[height=2.53in]{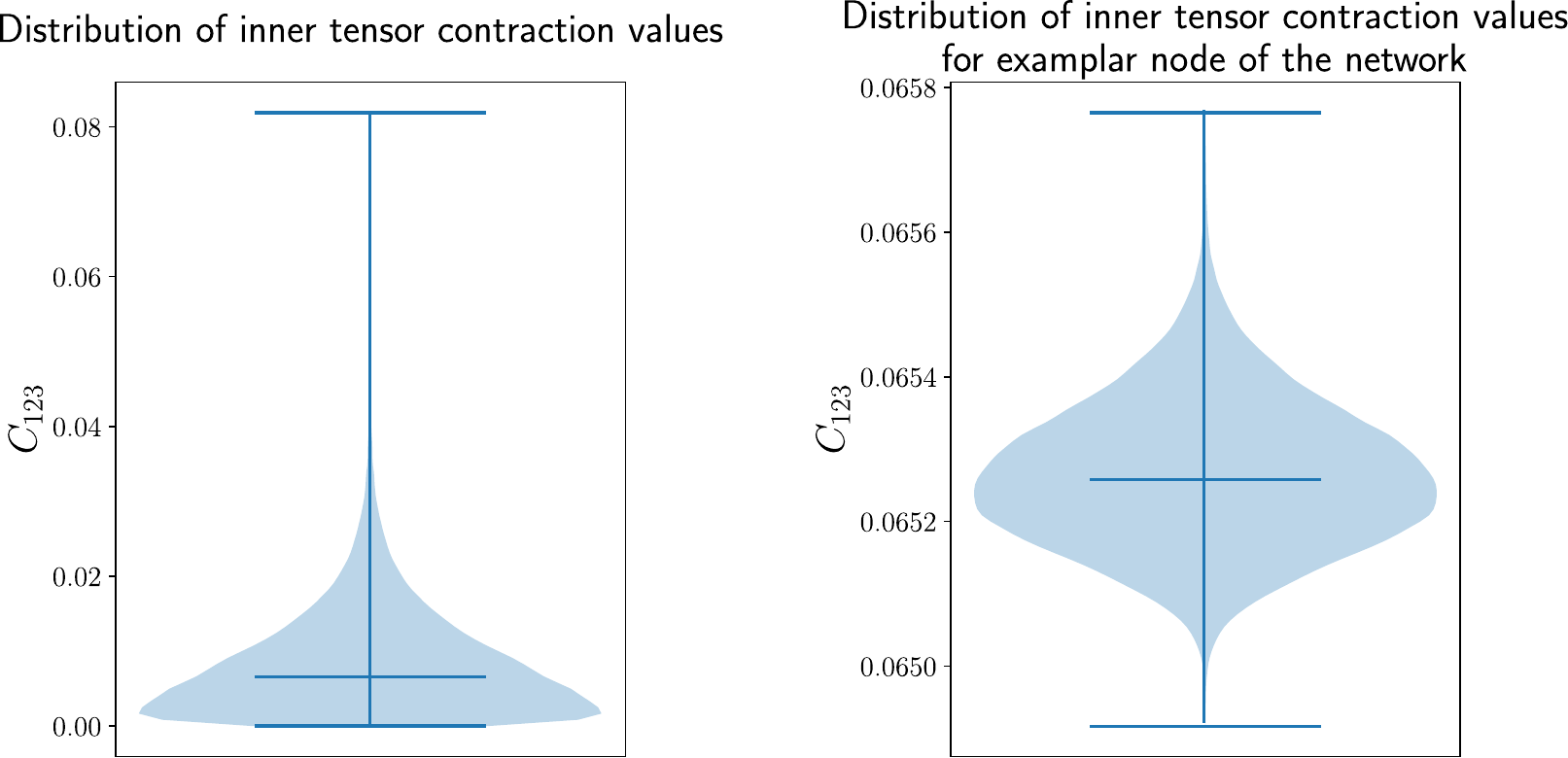}}
    \caption{Violin plots presenting values of contractions of inner tensors with the bulk operator as in Figure \ref{fig:3-point_corr}(b). \textbf{(a)} The distribution obtained from $10^5$ random bulk operators and random tensor building blocks,  drawn as in Figure \ref{fig:distribution_delta_random}. 
    \textbf{(b)} The distribution obtained from $10^5$ samples with tensor constructed form unitaries \eqref{blue_u_example}, \eqref{red_u_example1} with $a = 0.302$, $b = 0.817$ corresponding to minimal $\Delta$ in Figure \ref{fig:distribution_within_families}(c).
    In each case, we draw the bulk operators from Hermitian, positive semi-definite matrices with trace $1$ and considered triples of appropriate left normalized generalized eigenvectors $\la\mathbf{v_L}|$, as discussed above. The maximal, minimal and median values are highlighted.}
    \label{fig:distribution_C}
\end{figure}
We note that the approximations of both $|\mathbf{v_R}\ra$ and $\la\mathbf{v_L}|$ by SVD are normalized, which is not the case in Jordan decomposition. Thus we lost the proportionality factor which cannot be re-obtained without a knowledge of the size $k$ of the largest Jordan block corresponding to $\lambda_2$. However, since there can be more than one block to this eigenvalue, one cannot determine $k$ by simply calculating the multiplicity of generalized eigenvectors of $M$. Thus we return to the problem of performing Jordan decomposition numerically.

To overcome this difficulty one could explicitly contract all the path's nodes and the inner node and then by performing appropriate numerical fits obtain the values of all parameters of interest. However, we find this approach unsatisfactory, since it either doesn't nicely describe a large $n$ limit or is sensitive to numerical errors of multiple contractions, if one considers large values of $n$. 

In Figure \ref{fig:distribution_C} we present numerical explorations of possible values for inner tensor contraction values. To bound the range of obtained values, we restricted the bulk operator to \\ be Hermitian, positive semi-definite with trace $1$. In each scenario, we considered triples of appropriate left normalized generalized eigenvectors $\la\mathbf{v_L}|$. We note that resulting values are real up to numerical precision, which is in accordance with the fact that the original quantity, expectation value for a mapped hermitian operator, must be real as well.

\bibliography{Main}

\begin{thebibliography}{43}%
\makeatletter
\providecommand \@ifxundefined [1]{%
 \@ifx{#1\undefined}
}%
\providecommand \@ifnum [1]{%
 \ifnum #1\expandafter \@firstoftwo
 \else \expandafter \@secondoftwo
 \fi
}%
\providecommand \@ifx [1]{%
 \ifx #1\expandafter \@firstoftwo
 \else \expandafter \@secondoftwo
 \fi
}%
\providecommand \natexlab [1]{#1}%
\providecommand \enquote  [1]{``#1''}%
\providecommand \bibnamefont  [1]{#1}%
\providecommand \bibfnamefont [1]{#1}%
\providecommand \citenamefont [1]{#1}%
\providecommand \href@noop [0]{\@secondoftwo}%
\providecommand \href [0]{\begingroup \@sanitize@url \@href}%
\providecommand \@href[1]{\@@startlink{#1}\@@href}%
\providecommand \@@href[1]{\endgroup#1\@@endlink}%
\providecommand \@sanitize@url [0]{\catcode `\\12\catcode `\$12\catcode `\&12\catcode `\#12\catcode `\^12\catcode `\_12\catcode `\%12\relax}%
\providecommand \@@startlink[1]{}%
\providecommand \@@endlink[0]{}%
\providecommand \url  [0]{\begingroup\@sanitize@url \@url }%
\providecommand \@url [1]{\endgroup\@href {#1}{\urlprefix }}%
\providecommand \urlprefix  [0]{URL }%
\providecommand \Eprint [0]{\href }%
\providecommand \doibase [0]{https://doi.org/}%
\providecommand \selectlanguage [0]{\@gobble}%
\providecommand \bibinfo  [0]{\@secondoftwo}%
\providecommand \bibfield  [0]{\@secondoftwo}%
\providecommand \translation [1]{[#1]}%
\providecommand \BibitemOpen [0]{}%
\providecommand \bibitemStop [0]{}%
\providecommand \bibitemNoStop [0]{.\EOS\space}%
\providecommand \EOS [0]{\spacefactor3000\relax}%
\providecommand \BibitemShut  [1]{\csname bibitem#1\endcsname}%
\let\auto@bib@innerbib\@empty
\bibitem [{\citenamefont {Maldacena}(1998)}]{Maldacena1998}%
  \BibitemOpen
  \bibfield  {author} {\bibinfo {author} {\bibfnamefont {J.~M.}\ \bibnamefont {Maldacena}},\ }\bibfield  {title} {\bibinfo {title} {The large-n limit of superconformal field theories and supergravity},\ }\href {https://doi.org/10.4310/ATMP.1998.v2.n2.a1} {\bibfield  {journal} {\bibinfo  {journal} {ATMP}\ }\textbf {\bibinfo {volume} {2}},\ \bibinfo {pages} {231} (\bibinfo {year} {1998})}\BibitemShut {NoStop}%
\bibitem [{\citenamefont {Witten}(1998)}]{Witten1998}%
  \BibitemOpen
  \bibfield  {author} {\bibinfo {author} {\bibfnamefont {E.}~\bibnamefont {Witten}},\ }\bibfield  {title} {\bibinfo {title} {Anti de {S}itter space and holography},\ }\href {https://doi.org/10.4310/ATMP.1998.v2.n2.a2} {\bibfield  {journal} {\bibinfo  {journal} {ATMP}\ }\textbf {\bibinfo {volume} {2}},\ \bibinfo {pages} {253} (\bibinfo {year} {1998})}\BibitemShut {NoStop}%
\bibitem [{\citenamefont {Jahn}\ and\ \citenamefont {Eisert}(2021)}]{Jahn2021}%
  \BibitemOpen
  \bibfield  {author} {\bibinfo {author} {\bibfnamefont {A.}~\bibnamefont {Jahn}}\ and\ \bibinfo {author} {\bibfnamefont {J.}~\bibnamefont {Eisert}},\ }\bibfield  {title} {\bibinfo {title} {Holographic tensor network models and quantum error correction: a topical review},\ }\href {https://doi.org/10.1088/2058-9565/ac0293} {\bibfield  {journal} {\bibinfo  {journal} {QST}\ }\textbf {\bibinfo {volume} {6}},\ \bibinfo {pages} {033002} (\bibinfo {year} {2021})}\BibitemShut {NoStop}%
\bibitem [{\citenamefont {Pastawski}\ \emph {et~al.}(2015)\citenamefont {Pastawski}, \citenamefont {Yoshida}, \citenamefont {Harlow},\ and\ \citenamefont {Preskill}}]{HaPPy_code}%
  \BibitemOpen
  \bibfield  {author} {\bibinfo {author} {\bibfnamefont {F.}~\bibnamefont {Pastawski}}, \bibinfo {author} {\bibfnamefont {B.}~\bibnamefont {Yoshida}}, \bibinfo {author} {\bibfnamefont {D.}~\bibnamefont {Harlow}},\ and\ \bibinfo {author} {\bibfnamefont {J.}~\bibnamefont {Preskill}},\ }\bibfield  {title} {\bibinfo {title} {Holographic quantum error-correcting codes: toy models for the bulk/boundary correspondence},\ }\href {https://doi.org/10.1007/jhep06(2015)149} {\bibfield  {journal} {\bibinfo  {journal} {JHEP}\ }\textbf {\bibinfo {volume} {06}}\bibinfo  {number} { (06)},\ \bibinfo {pages} {149}}\BibitemShut {NoStop}%
\bibitem [{\citenamefont {Hein}\ \emph {et~al.}(2006)\citenamefont {Hein}, \citenamefont {Dür}, \citenamefont {Eisert}, \citenamefont {Raussendorf}, \citenamefont {Nest},\ and\ \citenamefont {Briegel}}]{perfect_tensors}%
  \BibitemOpen
\bibfield  {number} {  }\bibfield  {author} {\bibinfo {author} {\bibfnamefont {M.}~\bibnamefont {Hein}}, \bibinfo {author} {\bibfnamefont {W.}~\bibnamefont {Dür}}, \bibinfo {author} {\bibfnamefont {J.}~\bibnamefont {Eisert}}, \bibinfo {author} {\bibfnamefont {R.}~\bibnamefont {Raussendorf}}, \bibinfo {author} {\bibfnamefont {M.}~\bibnamefont {Nest}},\ and\ \bibinfo {author} {\bibfnamefont {H.}~\bibnamefont {Briegel}},\ }\bibfield  {title} {\bibinfo {title} {Entanglement in graph states and its applications},\ }\href {https://doi.org/10.3254/1-58603-660-2-115} {\bibfield  {journal} {\bibinfo  {journal} {Proceedings of the International School of Physics "Enrico Fermi"}\ }\textbf {\bibinfo {volume} {162}},\ \bibinfo {pages} {115} (\bibinfo {year} {2006})}\BibitemShut {NoStop}%
\bibitem [{\citenamefont {Goyeneche}\ \emph {et~al.}(2015)\citenamefont {Goyeneche}, \citenamefont {Alsina}, \citenamefont {Latorre}, \citenamefont {Riera},\ and\ \citenamefont {\.Zyczkowski}}]{Goyeneche:2015fda}%
  \BibitemOpen
  \bibfield  {author} {\bibinfo {author} {\bibfnamefont {D.}~\bibnamefont {Goyeneche}}, \bibinfo {author} {\bibfnamefont {D.}~\bibnamefont {Alsina}}, \bibinfo {author} {\bibfnamefont {J.~I.}\ \bibnamefont {Latorre}}, \bibinfo {author} {\bibfnamefont {A.}~\bibnamefont {Riera}},\ and\ \bibinfo {author} {\bibfnamefont {K.}~\bibnamefont {\.Zyczkowski}},\ }\bibfield  {title} {\bibinfo {title} {{Absolutely maximally entangled states, combinatorial designs, and multiunitary matrices}},\ }\href {https://doi.org/10.1103/PhysRevA.92.032316} {\bibfield  {journal} {\bibinfo  {journal} {Phys. Rev. A}\ }\textbf {\bibinfo {volume} {92}},\ \bibinfo {pages} {032316} (\bibinfo {year} {2015})}\BibitemShut {NoStop}%
\bibitem [{\citenamefont {Goyeneche}\ \emph {et~al.}(2018)\citenamefont {Goyeneche}, \citenamefont {Raissi}, \citenamefont {Di~Martino},\ and\ \citenamefont {\ifmmode~\dot{Z}\else \.{Z}\fi{}yczkowski}}]{quantum_perfect_tensors}%
  \BibitemOpen
  \bibfield  {author} {\bibinfo {author} {\bibfnamefont {D.}~\bibnamefont {Goyeneche}}, \bibinfo {author} {\bibfnamefont {Z.}~\bibnamefont {Raissi}}, \bibinfo {author} {\bibfnamefont {S.}~\bibnamefont {Di~Martino}},\ and\ \bibinfo {author} {\bibfnamefont {K.}~\bibnamefont {\ifmmode~\dot{Z}\else \.{Z}\fi{}yczkowski}},\ }\bibfield  {title} {\bibinfo {title} {Entanglement and quantum combinatorial designs},\ }\href {https://doi.org/10.1103/PhysRevA.97.062326} {\bibfield  {journal} {\bibinfo  {journal} {Phys. Rev. A}\ }\textbf {\bibinfo {volume} {97}},\ \bibinfo {pages} {062326} (\bibinfo {year} {2018})}\BibitemShut {NoStop}%
\bibitem [{\citenamefont {Eisert}\ \emph {et~al.}(2010)\citenamefont {Eisert}, \citenamefont {Cramer},\ and\ \citenamefont {Plenio}}]{ECP10}%
  \BibitemOpen
  \bibfield  {author} {\bibinfo {author} {\bibfnamefont {J.}~\bibnamefont {Eisert}}, \bibinfo {author} {\bibfnamefont {M.}~\bibnamefont {Cramer}},\ and\ \bibinfo {author} {\bibfnamefont {M.~B.}\ \bibnamefont {Plenio}},\ }\bibfield  {title} {\bibinfo {title} {Colloquium: Area laws for the entanglement entropy},\ }\href {https://doi.org/10.1103/RevModPhys.82.277} {\bibfield  {journal} {\bibinfo  {journal} {Rev. Mod. Phys.}\ }\textbf {\bibinfo {volume} {82}},\ \bibinfo {pages} {277} (\bibinfo {year} {2010})}\BibitemShut {NoStop}%
\bibitem [{\citenamefont {Ryu}\ and\ \citenamefont {Takayanagi}(2006)}]{Ryu2006}%
  \BibitemOpen
  \bibfield  {author} {\bibinfo {author} {\bibfnamefont {S.}~\bibnamefont {Ryu}}\ and\ \bibinfo {author} {\bibfnamefont {T.}~\bibnamefont {Takayanagi}},\ }\bibfield  {title} {\bibinfo {title} {Holographic derivation of entanglement entropy from the anti--de {S}itter space/conformal field theory correspondence},\ }\href {https://doi.org/10.1103/PhysRevLett.96.181602} {\bibfield  {journal} {\bibinfo  {journal} {Phys. Rev. Lett.}\ }\textbf {\bibinfo {volume} {96}},\ \bibinfo {pages} {181602} (\bibinfo {year} {2006})}\BibitemShut {NoStop}%
\bibitem [{\citenamefont {Cao}\ \emph {et~al.}(2022)\citenamefont {Cao}, \citenamefont {Pollack},\ and\ \citenamefont {Wang}}]{ChunJun2022}%
  \BibitemOpen
  \bibfield  {author} {\bibinfo {author} {\bibfnamefont {C.}~\bibnamefont {Cao}}, \bibinfo {author} {\bibfnamefont {J.}~\bibnamefont {Pollack}},\ and\ \bibinfo {author} {\bibfnamefont {Y.}~\bibnamefont {Wang}},\ }\bibfield  {title} {\bibinfo {title} {Hyperinvariant multiscale entanglement renormalization ansatz: Approximate holographic error correction codes with power-law correlations},\ }\href {https://doi.org/10.1103/physrevd.105.026018} {\bibfield  {journal} {\bibinfo  {journal} {Phys. Rev. D}\ }\textbf {\bibinfo {volume} {105}},\ \bibinfo {pages} {026018} (\bibinfo {year} {2022})}\BibitemShut {NoStop}%
\bibitem [{\citenamefont {Bhattacharyya}\ \emph {et~al.}(2016)\citenamefont {Bhattacharyya}, \citenamefont {Gao}, \citenamefont {Hung},\ and\ \citenamefont {Liu}}]{bhattacharyya2016exploring}%
  \BibitemOpen
  \bibfield  {author} {\bibinfo {author} {\bibfnamefont {A.}~\bibnamefont {Bhattacharyya}}, \bibinfo {author} {\bibfnamefont {Z.-S.}\ \bibnamefont {Gao}}, \bibinfo {author} {\bibfnamefont {L.-Y.}\ \bibnamefont {Hung}},\ and\ \bibinfo {author} {\bibfnamefont {S.-N.}\ \bibnamefont {Liu}},\ }\bibfield  {title} {\bibinfo {title} {Exploring the tensor networks/ads correspondence},\ }\href@noop {} {\bibfield  {journal} {\bibinfo  {journal} {Journal of High Energy Physics}\ }\textbf {\bibinfo {volume} {2016}},\ \bibinfo {pages} {1} (\bibinfo {year} {2016})}\BibitemShut {NoStop}%
\bibitem [{\citenamefont {Bhattacharyya}\ \emph {et~al.}(2018)\citenamefont {Bhattacharyya}, \citenamefont {Hung}, \citenamefont {Lei},\ and\ \citenamefont {Li}}]{bhattacharyya2018tensor}%
  \BibitemOpen
  \bibfield  {author} {\bibinfo {author} {\bibfnamefont {A.}~\bibnamefont {Bhattacharyya}}, \bibinfo {author} {\bibfnamefont {L.-Y.}\ \bibnamefont {Hung}}, \bibinfo {author} {\bibfnamefont {Y.}~\bibnamefont {Lei}},\ and\ \bibinfo {author} {\bibfnamefont {W.}~\bibnamefont {Li}},\ }\bibfield  {title} {\bibinfo {title} {Tensor network and (p-adic) ads/cft},\ }\href@noop {} {\bibfield  {journal} {\bibinfo  {journal} {Journal of High Energy Physics}\ }\textbf {\bibinfo {volume} {2018}} (\bibinfo {year} {2018})}\BibitemShut {NoStop}%
\bibitem [{\citenamefont {Pfeifer}\ \emph {et~al.}(2009)\citenamefont {Pfeifer}, \citenamefont {Evenbly},\ and\ \citenamefont {Vidal}}]{Pfiefer2009}%
  \BibitemOpen
  \bibfield  {author} {\bibinfo {author} {\bibfnamefont {R.~N.~C.}\ \bibnamefont {Pfeifer}}, \bibinfo {author} {\bibfnamefont {G.}~\bibnamefont {Evenbly}},\ and\ \bibinfo {author} {\bibfnamefont {G.}~\bibnamefont {Vidal}},\ }\bibfield  {title} {\bibinfo {title} {Entanglement renormalization, scale invariance, and quantum criticality},\ }\href {https://doi.org/10.1103/PhysRevA.79.040301} {\bibfield  {journal} {\bibinfo  {journal} {Phys. Rev. A}\ }\textbf {\bibinfo {volume} {79}},\ \bibinfo {pages} {040301} (\bibinfo {year} {2009})}\BibitemShut {NoStop}%
\bibitem [{\citenamefont {Holzhey}\ \emph {et~al.}(1994)\citenamefont {Holzhey}, \citenamefont {Larsen},\ and\ \citenamefont {Wilczek}}]{Holzhey1994}%
  \BibitemOpen
  \bibfield  {author} {\bibinfo {author} {\bibfnamefont {C.}~\bibnamefont {Holzhey}}, \bibinfo {author} {\bibfnamefont {F.}~\bibnamefont {Larsen}},\ and\ \bibinfo {author} {\bibfnamefont {F.}~\bibnamefont {Wilczek}},\ }\bibfield  {title} {\bibinfo {title} {Geometric and renormalized entropy in conformal field theory},\ }\href {https://doi.org/10.1016/0550-3213(94)90402-2} {\bibfield  {journal} {\bibinfo  {journal} {Nucl. Phys. B}\ }\textbf {\bibinfo {volume} {424}},\ \bibinfo {pages} {443} (\bibinfo {year} {1994})}\BibitemShut {NoStop}%
\bibitem [{\citenamefont {Evenbly}(2017)}]{Evenbly2017}%
  \BibitemOpen
  \bibfield  {author} {\bibinfo {author} {\bibfnamefont {G.}~\bibnamefont {Evenbly}},\ }\bibfield  {title} {\bibinfo {title} {Hyperinvariant tensor networks and holography},\ }\href {https://doi.org/10.1103/physrevlett.119.141602} {\bibfield  {journal} {\bibinfo  {journal} {Phys. Rev. Lett.}\ }\textbf {\bibinfo {volume} {119}},\ \bibinfo {pages} {141602} (\bibinfo {year} {2017})}\BibitemShut {NoStop}%
\bibitem [{\citenamefont {Jahn}\ \emph {et~al.}(2022{\natexlab{a}})\citenamefont {Jahn}, \citenamefont {Zimborás},\ and\ \citenamefont {Eisert}}]{Jahn2022}%
  \BibitemOpen
  \bibfield  {author} {\bibinfo {author} {\bibfnamefont {A.}~\bibnamefont {Jahn}}, \bibinfo {author} {\bibfnamefont {Z.}~\bibnamefont {Zimborás}},\ and\ \bibinfo {author} {\bibfnamefont {J.}~\bibnamefont {Eisert}},\ }\bibfield  {title} {\bibinfo {title} {Tensor network models of {AdS/qCFT}},\ }\href {https://doi.org/10.22331/q-2022-02-03-643} {\bibfield  {journal} {\bibinfo  {journal} {Quantum}\ }\textbf {\bibinfo {volume} {6}},\ \bibinfo {pages} {643} (\bibinfo {year} {2022}{\natexlab{a}})}\BibitemShut {NoStop}%
\bibitem [{\citenamefont {Steinberg}\ and\ \citenamefont {Prior}(2022)}]{Steinberg2022}%
  \BibitemOpen
  \bibfield  {author} {\bibinfo {author} {\bibfnamefont {M.}~\bibnamefont {Steinberg}}\ and\ \bibinfo {author} {\bibfnamefont {J.}~\bibnamefont {Prior}},\ }\bibfield  {title} {\bibinfo {title} {Conformal properties of hyperinvariant tensor networks},\ }\href {https://doi.org/10.1038/s41598-021-04375-5} {\bibfield  {journal} {\bibinfo  {journal} {Sci. Rep.}\ }\textbf {\bibinfo {volume} {12}},\ \bibinfo {pages} {532} (\bibinfo {year} {2022})}\BibitemShut {NoStop}%
\bibitem [{\citenamefont {Akila}\ \emph {et~al.}(2016)\citenamefont {Akila}, \citenamefont {Waltner}, \citenamefont {Gutkin},\ and\ \citenamefont {Guhr}}]{AWGG16}%
  \BibitemOpen
  \bibfield  {author} {\bibinfo {author} {\bibfnamefont {M.}~\bibnamefont {Akila}}, \bibinfo {author} {\bibfnamefont {D.}~\bibnamefont {Waltner}}, \bibinfo {author} {\bibfnamefont {B.}~\bibnamefont {Gutkin}},\ and\ \bibinfo {author} {\bibfnamefont {T.}~\bibnamefont {Guhr}},\ }\bibfield  {title} {\bibinfo {title} {Particle-time duality in the kicked {I}sing spin chain},\ }\href {https://doi.org/10.1088/1751-8113/49/37/375101} {\bibfield  {journal} {\bibinfo  {journal} {J. Phys. A:}\ }\textbf {\bibinfo {volume} {49}},\ \bibinfo {pages} {375101} (\bibinfo {year} {2016})}\BibitemShut {NoStop}%
\bibitem [{\citenamefont {Bertini}\ \emph {et~al.}(2019)\citenamefont {Bertini}, \citenamefont {Kos},\ and\ \citenamefont {Prosen}}]{BKP19}%
  \BibitemOpen
  \bibfield  {author} {\bibinfo {author} {\bibfnamefont {B.}~\bibnamefont {Bertini}}, \bibinfo {author} {\bibfnamefont {P.}~\bibnamefont {Kos}},\ and\ \bibinfo {author} {\bibfnamefont {T.}~\bibnamefont {Prosen}},\ }\bibfield  {title} {\bibinfo {title} {Exact correlation functions for dual-unitary lattice models in $1+1$ dimensions},\ }\href {https://doi.org/10.1103/PhysRevLett.123.210601} {\bibfield  {journal} {\bibinfo  {journal} {Phys. Rev. Lett.}\ }\textbf {\bibinfo {volume} {123}},\ \bibinfo {pages} {210601} (\bibinfo {year} {2019})}\BibitemShut {NoStop}%
\bibitem [{\citenamefont {Brahmachari}\ \emph {et~al.}(2024{\natexlab{a}})\citenamefont {Brahmachari}, \citenamefont {Rajmohan}, \citenamefont {Rather},\ and\ \citenamefont {Lakshminarayan}}]{BRRL24}%
  \BibitemOpen
  \bibfield  {author} {\bibinfo {author} {\bibfnamefont {S.}~\bibnamefont {Brahmachari}}, \bibinfo {author} {\bibfnamefont {R.~N.}\ \bibnamefont {Rajmohan}}, \bibinfo {author} {\bibfnamefont {S.~A.}\ \bibnamefont {Rather}},\ and\ \bibinfo {author} {\bibfnamefont {A.}~\bibnamefont {Lakshminarayan}},\ }\bibfield  {title} {\bibinfo {title} {Dual unitaries as maximizers of the distance to local product gates},\ }\href {https://doi.org/10.1103/PhysRevA.109.022610} {\bibfield  {journal} {\bibinfo  {journal} {Phys. Rev. A}\ }\textbf {\bibinfo {volume} {109}},\ \bibinfo {pages} {022610} (\bibinfo {year} {2024}{\natexlab{a}})}\BibitemShut {NoStop}%
\bibitem [{\citenamefont {Harris}\ \emph {et~al.}(2018)\citenamefont {Harris}, \citenamefont {McMahon}, \citenamefont {Brennen},\ and\ \citenamefont {Stace}}]{harris2018calderbanksteaneshor}%
  \BibitemOpen
  \bibfield  {author} {\bibinfo {author} {\bibfnamefont {R.~J.}\ \bibnamefont {Harris}}, \bibinfo {author} {\bibfnamefont {N.~A.}\ \bibnamefont {McMahon}}, \bibinfo {author} {\bibfnamefont {G.~K.}\ \bibnamefont {Brennen}},\ and\ \bibinfo {author} {\bibfnamefont {T.~M.}\ \bibnamefont {Stace}},\ }\bibfield  {title} {\bibinfo {title} {Calderbank-steane-shor holographic quantum error correcting codes},\ }\href@noop {} {\bibfield  {journal} {\bibinfo  {journal} {arXiv preprint arXiv: 1806.06472}\ } (\bibinfo {year} {2018})}\BibitemShut {NoStop}%
\bibitem [{\citenamefont {Wang}(2021)}]{planar_k_uniform_states}%
  \BibitemOpen
  \bibfield  {author} {\bibinfo {author} {\bibfnamefont {Y.}~\bibnamefont {Wang}},\ }\bibfield  {title} {\bibinfo {title} {Planar k-uniform states: a generalization of planar maximally entangled states},\ }\href {https://doi.org/10.1007/S11128-021-03204-Y} {\bibfield  {journal} {\bibinfo  {journal} {Quantum Inf. Process.}\ }\textbf {\bibinfo {volume} {20}},\ \bibinfo {pages} {1} (\bibinfo {year} {2021})}\BibitemShut {NoStop}%
\bibitem [{\citenamefont {Steinberg}\ \emph {et~al.}(2023)\citenamefont {Steinberg}, \citenamefont {Feld},\ and\ \citenamefont {Jahn}}]{Steinberg2023}%
  \BibitemOpen
  \bibfield  {author} {\bibinfo {author} {\bibfnamefont {M.}~\bibnamefont {Steinberg}}, \bibinfo {author} {\bibfnamefont {S.}~\bibnamefont {Feld}},\ and\ \bibinfo {author} {\bibfnamefont {A.}~\bibnamefont {Jahn}},\ }\bibfield  {title} {\bibinfo {title} {Holographic codes from hyperinvariant tensor networks},\ }\bibfield  {journal} {\bibinfo  {journal} {Nat. Commun.}\ }\textbf {\bibinfo {volume} {14}},\ \href {https://doi.org/10.1038/s41467-023-42743-z} {10.1038/s41467-023-42743-z} (\bibinfo {year} {2023})\BibitemShut {NoStop}%
\bibitem [{\citenamefont {Huber}\ \emph {et~al.}(2018)\citenamefont {Huber}, \citenamefont {Eltschka}, \citenamefont {Siewert},\ and\ \citenamefont {G\"uhne}}]{HESGW18}%
  \BibitemOpen
  \bibfield  {author} {\bibinfo {author} {\bibfnamefont {F.}~\bibnamefont {Huber}}, \bibinfo {author} {\bibfnamefont {C.}~\bibnamefont {Eltschka}}, \bibinfo {author} {\bibfnamefont {J.}~\bibnamefont {Siewert}},\ and\ \bibinfo {author} {\bibfnamefont {O.}~\bibnamefont {G\"uhne}},\ }\bibfield  {title} {\bibinfo {title} {{Bounds on absolutely maximally entangled states from shadow inequalities, and the quantum MacWilliams identity}},\ }\href {https://doi.org/10.1088/1751-8121/aaade5} {\bibfield  {journal} {\bibinfo  {journal} {J. Phys. A}\ }\textbf {\bibinfo {volume} {51}},\ \bibinfo {pages} {175301} (\bibinfo {year} {2018})}\BibitemShut {NoStop}%
\bibitem [{\citenamefont {Huber}\ and\ \citenamefont {Wyderka}(2019)}]{HW19b}%
  \BibitemOpen
  \bibfield  {author} {\bibinfo {author} {\bibfnamefont {F.}~\bibnamefont {Huber}}\ and\ \bibinfo {author} {\bibfnamefont {N.}~\bibnamefont {Wyderka}},\ }\bibfield  {title} {\bibinfo {title} {Table of {AME} states \& perfect tensors},\ }\href@noop {} {\bibfield  {journal} {\bibinfo  {journal} {Available online at:}\ } (\bibinfo {year} {2019})},\ \Eprint {https://arxiv.org/abs/https://tp.nt.uni-siegen.de/ame/ame.html} {https://tp.nt.uni-siegen.de/ame/ame.html} \BibitemShut {NoStop}%
\bibitem [{\citenamefont {Steinberg}\ \emph {et~al.}(2024)\citenamefont {Steinberg}, \citenamefont {Fan}, \citenamefont {Harris}, \citenamefont {Elkouss}, \citenamefont {Feld},\ and\ \citenamefont {Jahn}}]{Steinberg:2024ack}%
  \BibitemOpen
  \bibfield  {author} {\bibinfo {author} {\bibfnamefont {M.}~\bibnamefont {Steinberg}}, \bibinfo {author} {\bibfnamefont {J.}~\bibnamefont {Fan}}, \bibinfo {author} {\bibfnamefont {R.~J.}\ \bibnamefont {Harris}}, \bibinfo {author} {\bibfnamefont {D.}~\bibnamefont {Elkouss}}, \bibinfo {author} {\bibfnamefont {S.}~\bibnamefont {Feld}},\ and\ \bibinfo {author} {\bibfnamefont {A.}~\bibnamefont {Jahn}},\ }\bibfield  {title} {\bibinfo {title} {Far from perfect: Quantum error correction with (hyperinvariant) {E}venbly codes},\ }\bibfield  {journal} {\bibinfo  {journal} {arXiv preprint: 2407.11926}\ }\href {https://doi.org/https://doi.org/10.48550/arXiv.2407.11926} {https://doi.org/10.48550/arXiv.2407.11926} (\bibinfo {year} {2024})\BibitemShut {NoStop}%
\bibitem [{\citenamefont {Jahn}\ \emph {et~al.}(2022{\natexlab{b}})\citenamefont {Jahn}, \citenamefont {Gluza}, \citenamefont {Verhoeven}, \citenamefont {Singhd},\ and\ \citenamefont {Eisert}}]{Jahn2022_1}%
  \BibitemOpen
  \bibfield  {author} {\bibinfo {author} {\bibfnamefont {A.}~\bibnamefont {Jahn}}, \bibinfo {author} {\bibfnamefont {M.}~\bibnamefont {Gluza}}, \bibinfo {author} {\bibfnamefont {C.}~\bibnamefont {Verhoeven}}, \bibinfo {author} {\bibfnamefont {S.}~\bibnamefont {Singhd}},\ and\ \bibinfo {author} {\bibfnamefont {J.}~\bibnamefont {Eisert}},\ }\bibfield  {title} {\bibinfo {title} {Boundary theories of critical matchgate tensor networks},\ }\href {https://doi.org/10.1007/JHEP04(2022)111} {\bibfield  {journal} {\bibinfo  {journal} {JHEP}\ }\textbf {\bibinfo {volume} {111}},\ \bibinfo {pages} {110}}\BibitemShut {NoStop}%
\bibitem [{\citenamefont {Evenbly}\ and\ \citenamefont {Vidal}(2014)}]{Evenbly2014}%
  \BibitemOpen
  \bibfield  {author} {\bibinfo {author} {\bibfnamefont {G.}~\bibnamefont {Evenbly}}\ and\ \bibinfo {author} {\bibfnamefont {G.}~\bibnamefont {Vidal}},\ }\bibfield  {title} {\bibinfo {title} {Algorithms for entanglement renormalization: Boundaries, impurities and interfaces},\ }\href {https://doi.org/10.1007/s10955-014-0983-1} {\bibfield  {journal} {\bibinfo  {journal} {J. Stat. Phys.}\ }\textbf {\bibinfo {volume} {157}},\ \bibinfo {pages} {931} (\bibinfo {year} {2014})}\BibitemShut {NoStop}%
\bibitem [{\citenamefont {Evenbly}\ and\ \citenamefont {Vidal}(2011)}]{Evenbly2011}%
  \BibitemOpen
  \bibfield  {author} {\bibinfo {author} {\bibfnamefont {G.}~\bibnamefont {Evenbly}}\ and\ \bibinfo {author} {\bibfnamefont {G.}~\bibnamefont {Vidal}},\ }\bibfield  {title} {\bibinfo {title} {Tensor network states and geometry},\ }\href {https://doi.org/10.1007/s10955-011-0237-4} {\bibfield  {journal} {\bibinfo  {journal} {J. Stat. Phys.}\ }\textbf {\bibinfo {volume} {145}},\ \bibinfo {pages} {891–918} (\bibinfo {year} {2011})}\BibitemShut {NoStop}%
\bibitem [{\citenamefont {Jahn}\ \emph {et~al.}(2019)\citenamefont {Jahn}, \citenamefont {Gluza}, \citenamefont {Pastawski},\ and\ \citenamefont {Eisert}}]{Jahn:2019nmz}%
  \BibitemOpen
  \bibfield  {author} {\bibinfo {author} {\bibfnamefont {A.}~\bibnamefont {Jahn}}, \bibinfo {author} {\bibfnamefont {M.}~\bibnamefont {Gluza}}, \bibinfo {author} {\bibfnamefont {F.}~\bibnamefont {Pastawski}},\ and\ \bibinfo {author} {\bibfnamefont {J.}~\bibnamefont {Eisert}},\ }\bibfield  {title} {\bibinfo {title} {{Majorana dimers and holographic quantum error-correcting codes}},\ }\href {https://doi.org/10.1103/PhysRevResearch.1.033079} {\bibfield  {journal} {\bibinfo  {journal} {Phys. Rev. Res.}\ }\textbf {\bibinfo {volume} {1}},\ \bibinfo {pages} {033079} (\bibinfo {year} {2019})}\BibitemShut {NoStop}%
\bibitem [{\citenamefont {Jahn}\ \emph {et~al.}(2020)\citenamefont {Jahn}, \citenamefont {Zimbor\'as},\ and\ \citenamefont {Eisert}}]{Central_charges_and_scaling}%
  \BibitemOpen
  \bibfield  {author} {\bibinfo {author} {\bibfnamefont {A.}~\bibnamefont {Jahn}}, \bibinfo {author} {\bibfnamefont {Z.}~\bibnamefont {Zimbor\'as}},\ and\ \bibinfo {author} {\bibfnamefont {J.}~\bibnamefont {Eisert}},\ }\bibfield  {title} {\bibinfo {title} {{Central charges of aperiodic holographic tensor network models}},\ }\href {https://doi.org/10.1103/PhysRevA.102.042407} {\bibfield  {journal} {\bibinfo  {journal} {Phys. Rev. A}\ }\textbf {\bibinfo {volume} {102}},\ \bibinfo {pages} {042407} (\bibinfo {year} {2020})}\BibitemShut {NoStop}%
\bibitem [{\citenamefont {Ginsparg}(1988)}]{Ginsparg:1988ui}%
  \BibitemOpen
  \bibfield  {author} {\bibinfo {author} {\bibfnamefont {P.~H.}\ \bibnamefont {Ginsparg}},\ }\bibfield  {title} {\bibinfo {title} {Applied conformal field theory},\ }in\ \href@noop {} {\emph {\bibinfo {booktitle} {{Les Houches Summer School in Theoretical Physics: Fields, Strings, Critical Phenomena}}}}\ (\bibinfo {year} {1988})\ \Eprint {https://arxiv.org/abs/hep-th/9108028} {arXiv:hep-th/9108028} \BibitemShut {NoStop}%
\bibitem [{\citenamefont {Conway}(2007)}]{Conway2007}%
  \BibitemOpen
  \bibfield  {author} {\bibinfo {author} {\bibfnamefont {J.~B.}\ \bibnamefont {Conway}},\ }\href {https://doi.org/10.1007/978-1-4757-4383-8} {\emph {\bibinfo {title} {A Course in Functional Analysis}}}\ (\bibinfo  {publisher} {Springer, New York},\ \bibinfo {year} {2007})\BibitemShut {NoStop}%
\bibitem [{\citenamefont {Rather}\ \emph {et~al.}(2022)\citenamefont {Rather}, \citenamefont {Burchardt}, \citenamefont {Bruzda}, \citenamefont {Rajchel-Mieldzio\ifmmode~\acute{c}\else \'{c}\fi{}}, \citenamefont {Lakshminarayan},\ and\ \citenamefont {\ifmmode~\dot{Z}\else \.{Z}\fi{}yczkowski}}]{36_officers_of_Karol}%
  \BibitemOpen
  \bibfield  {author} {\bibinfo {author} {\bibfnamefont {S.~A.}\ \bibnamefont {Rather}}, \bibinfo {author} {\bibfnamefont {A.}~\bibnamefont {Burchardt}}, \bibinfo {author} {\bibfnamefont {W.}~\bibnamefont {Bruzda}}, \bibinfo {author} {\bibfnamefont {G.}~\bibnamefont {Rajchel-Mieldzio\ifmmode~\acute{c}\else \'{c}\fi{}}}, \bibinfo {author} {\bibfnamefont {A.}~\bibnamefont {Lakshminarayan}},\ and\ \bibinfo {author} {\bibfnamefont {K.}~\bibnamefont {\ifmmode~\dot{Z}\else \.{Z}\fi{}yczkowski}},\ }\bibfield  {title} {\bibinfo {title} {Thirty-six entangled officers of {E}uler: Quantum solution to a classically impossible problem},\ }\href {https://doi.org/10.1103/PhysRevLett.128.080507} {\bibfield  {journal} {\bibinfo  {journal} {Phys. Rev. Lett.}\ }\textbf {\bibinfo {volume} {128}},\ \bibinfo {pages} {080507} (\bibinfo {year} {2022})}\BibitemShut {NoStop}%
\bibitem [{\citenamefont {Poland}\ \emph {et~al.}(2019)\citenamefont {Poland}, \citenamefont {Rychkov},\ and\ \citenamefont {Vichi}}]{Poland:2018epd}%
  \BibitemOpen
  \bibfield  {author} {\bibinfo {author} {\bibfnamefont {D.}~\bibnamefont {Poland}}, \bibinfo {author} {\bibfnamefont {S.}~\bibnamefont {Rychkov}},\ and\ \bibinfo {author} {\bibfnamefont {A.}~\bibnamefont {Vichi}},\ }\bibfield  {title} {\bibinfo {title} {{The Conformal Bootstrap: Theory, numerical techniques, and applications}},\ }\href {https://doi.org/10.1103/RevModPhys.91.015002} {\bibfield  {journal} {\bibinfo  {journal} {Rev. Mod. Phys.}\ }\textbf {\bibinfo {volume} {91}},\ \bibinfo {pages} {015002} (\bibinfo {year} {2019})}\BibitemShut {NoStop}%
\bibitem [{\citenamefont {Bistroń}\ \emph {et~al.}(2023)\citenamefont {Bistroń}, \citenamefont {Czartowski},\ and\ \citenamefont {Życzkowski}}]{bistron2023genuinely}%
  \BibitemOpen
  \bibfield  {author} {\bibinfo {author} {\bibfnamefont {R.}~\bibnamefont {Bistroń}}, \bibinfo {author} {\bibfnamefont {J.}~\bibnamefont {Czartowski}},\ and\ \bibinfo {author} {\bibfnamefont {K.}~\bibnamefont {Życzkowski}},\ }\bibfield  {title} {\bibinfo {title} {Genuinely quantum families of 2-unitary matrices},\ }\bibfield  {journal} {\bibinfo  {journal} {arXiv preprint: 2312.17719}\ }\href {https://doi.org/https://doi.org/10.48550/arXiv.2312.17719} {https://doi.org/10.48550/arXiv.2312.17719} (\bibinfo {year} {2023})\BibitemShut {NoStop}%
\bibitem [{\citenamefont {Kraus}\ and\ \citenamefont {Cirac}(2001)}]{2_qubit_gate}%
  \BibitemOpen
  \bibfield  {author} {\bibinfo {author} {\bibfnamefont {B.}~\bibnamefont {Kraus}}\ and\ \bibinfo {author} {\bibfnamefont {J.~I.}\ \bibnamefont {Cirac}},\ }\bibfield  {title} {\bibinfo {title} {Optimal creation of entanglement using a two-qubit gate},\ }\href {https://doi.org/10.1103/PhysRevA.63.062309} {\bibfield  {journal} {\bibinfo  {journal} {Phys. Rev. A}\ }\textbf {\bibinfo {volume} {63}},\ \bibinfo {pages} {062309} (\bibinfo {year} {2001})}\BibitemShut {NoStop}%
\bibitem [{\citenamefont {Hammerer}\ \emph {et~al.}(2002)\citenamefont {Hammerer}, \citenamefont {Vidal},\ and\ \citenamefont {Cirac}}]{HVC01}%
  \BibitemOpen
  \bibfield  {author} {\bibinfo {author} {\bibfnamefont {K.}~\bibnamefont {Hammerer}}, \bibinfo {author} {\bibfnamefont {G.}~\bibnamefont {Vidal}},\ and\ \bibinfo {author} {\bibfnamefont {J.~I.}\ \bibnamefont {Cirac}},\ }\bibfield  {title} {\bibinfo {title} {Characterization of nonlocal gates},\ }\href {https://doi.org/10.1103/PhysRevA.66.062321} {\bibfield  {journal} {\bibinfo  {journal} {Phys. Rev. A}\ }\textbf {\bibinfo {volume} {66}},\ \bibinfo {pages} {062321} (\bibinfo {year} {2002})}\BibitemShut {NoStop}%
\bibitem [{\citenamefont {Zhang}\ \emph {et~al.}(2003)\citenamefont {Zhang}, \citenamefont {Vala}, \citenamefont {Whaley},\ and\ \citenamefont {Sastry}}]{ZVSW03}%
  \BibitemOpen
  \bibfield  {author} {\bibinfo {author} {\bibfnamefont {J.}~\bibnamefont {Zhang}}, \bibinfo {author} {\bibfnamefont {J.}~\bibnamefont {Vala}}, \bibinfo {author} {\bibfnamefont {K.~B.}\ \bibnamefont {Whaley}},\ and\ \bibinfo {author} {\bibfnamefont {S.}~\bibnamefont {Sastry}},\ }\bibfield  {title} {\bibinfo {title} {{Geometric theory of nonlocal two-qubit operations}},\ }\href {https://doi.org/10.1103/PhysRevA.67.042313} {\bibfield  {journal} {\bibinfo  {journal} {Phys. Rev. A}\ }\textbf {\bibinfo {volume} {67}},\ \bibinfo {pages} {042313} (\bibinfo {year} {2003})}\BibitemShut {NoStop}%
\bibitem [{\citenamefont {Mandarino}\ \emph {et~al.}(2018)\citenamefont {Mandarino}, \citenamefont {Linowski},\ and\ \citenamefont {\ifmmode~\dot{Z}\else \.{Z}\fi{}yczkowski}}]{Weyl_chamber}%
  \BibitemOpen
  \bibfield  {author} {\bibinfo {author} {\bibfnamefont {A.}~\bibnamefont {Mandarino}}, \bibinfo {author} {\bibfnamefont {T.}~\bibnamefont {Linowski}},\ and\ \bibinfo {author} {\bibfnamefont {K.}~\bibnamefont {\ifmmode~\dot{Z}\else \.{Z}\fi{}yczkowski}},\ }\bibfield  {title} {\bibinfo {title} {Bipartite unitary gates and billiard dynamics in the {W}eyl chamber},\ }\href {https://doi.org/10.1103/PhysRevA.98.012335} {\bibfield  {journal} {\bibinfo  {journal} {Phys. Rev. A}\ }\textbf {\bibinfo {volume} {98}},\ \bibinfo {pages} {012335} (\bibinfo {year} {2018})}\BibitemShut {NoStop}%
\bibitem [{\citenamefont {Jonnadula}\ \emph {et~al.}(2017)\citenamefont {Jonnadula}, \citenamefont {Mandayam}, \citenamefont {\.Zyczkowski},\ and\ \citenamefont {Lakshminarayan}}]{jonnadula2017impact}%
  \BibitemOpen
  \bibfield  {author} {\bibinfo {author} {\bibfnamefont {B.}~\bibnamefont {Jonnadula}}, \bibinfo {author} {\bibfnamefont {P.}~\bibnamefont {Mandayam}}, \bibinfo {author} {\bibfnamefont {K.}~\bibnamefont {\.Zyczkowski}},\ and\ \bibinfo {author} {\bibfnamefont {A.}~\bibnamefont {Lakshminarayan}},\ }\bibfield  {title} {\bibinfo {title} {Impact of local dynamics on entangling power},\ }\href {https://doi.org/10.1103/PhysRevA.95.040302} {\bibfield  {journal} {\bibinfo  {journal} {Phys. Rev. A}\ }\textbf {\bibinfo {volume} {95}},\ \bibinfo {pages} {040302} (\bibinfo {year} {2017})}\BibitemShut {NoStop}%
\bibitem [{\citenamefont {Brahmachari}\ \emph {et~al.}(2024{\natexlab{b}})\citenamefont {Brahmachari}, \citenamefont {Rajmohan}, \citenamefont {Rather},\ and\ \citenamefont {Lakshminarayan}}]{Dual_u_max_dist}%
  \BibitemOpen
  \bibfield  {author} {\bibinfo {author} {\bibfnamefont {S.}~\bibnamefont {Brahmachari}}, \bibinfo {author} {\bibfnamefont {R.~N.}\ \bibnamefont {Rajmohan}}, \bibinfo {author} {\bibfnamefont {S.~A.}\ \bibnamefont {Rather}},\ and\ \bibinfo {author} {\bibfnamefont {A.}~\bibnamefont {Lakshminarayan}},\ }\bibfield  {title} {\bibinfo {title} {Dual unitaries as maximizers of the distance to local product gates},\ }\href {https://doi.org/10.1103/PhysRevA.109.022610} {\bibfield  {journal} {\bibinfo  {journal} {Phys. Rev. A}\ }\textbf {\bibinfo {volume} {109}},\ \bibinfo {pages} {022610} (\bibinfo {year} {2024}{\natexlab{b}})}\BibitemShut {NoStop}%
\bibitem [{\citenamefont {Nielsen}\ \emph {et~al.}(2003)\citenamefont {Nielsen}, \citenamefont {Dawson}, \citenamefont {Dodd}, \citenamefont {Gilchrist}, \citenamefont {Mortimer}, \citenamefont {Osborne}, \citenamefont {Bremner}, \citenamefont {Harrow},\ and\ \citenamefont {Hines}}]{PhysRevA.67.052301}%
  \BibitemOpen
  \bibfield  {author} {\bibinfo {author} {\bibfnamefont {M.~A.}\ \bibnamefont {Nielsen}}, \bibinfo {author} {\bibfnamefont {C.~M.}\ \bibnamefont {Dawson}}, \bibinfo {author} {\bibfnamefont {J.~L.}\ \bibnamefont {Dodd}}, \bibinfo {author} {\bibfnamefont {A.}~\bibnamefont {Gilchrist}}, \bibinfo {author} {\bibfnamefont {D.}~\bibnamefont {Mortimer}}, \bibinfo {author} {\bibfnamefont {T.~J.}\ \bibnamefont {Osborne}}, \bibinfo {author} {\bibfnamefont {M.~J.}\ \bibnamefont {Bremner}}, \bibinfo {author} {\bibfnamefont {A.~W.}\ \bibnamefont {Harrow}},\ and\ \bibinfo {author} {\bibfnamefont {A.}~\bibnamefont {Hines}},\ }\bibfield  {title} {\bibinfo {title} {Quantum dynamics as a physical resource},\ }\href {https://doi.org/10.1103/PhysRevA.67.052301} {\bibfield  {journal} {\bibinfo  {journal} {Phys. Rev. A}\ }\textbf {\bibinfo {volume} {67}},\ \bibinfo {pages} {052301} (\bibinfo {year} {2003})}\BibitemShut {NoStop}%
\end{thebibliography}%

\end{document}